\documentclass[12pt]{amsart}
\usepackage{graphicx}

\usepackage{srcltx}

\advance\textheight\topskip
\usepackage[OT2,T1]{fontenc}
\newcommand{\mifody}{%
  \renewcommand\rmdefault{wncyr}%
  \renewcommand\sfdefault{wncyss}%
  \renewcommand\encodingdefault{OT2}%
  \normalfont
  \selectfont}

\usepackage{amsmath,amsfonts,amsthm,amssymb,amscd}
\usepackage[mathscr]{eucal}
\allowdisplaybreaks
\usepackage{times}
\usepackage[curve,matrix,arrow]{xy}

\newcommand{\zeroRT}[1]{\Bigl[0^{\displaystyle^{#1}}\Bigr]}
\newcommand{\zeroLB}[1]{\Bigl[_{\displaystyle^{#1}}0\Bigr]}

\newcommand{\sg}{\mathrm{sg}}

\newcommand{\sbar}{\overline{s}}
\newcommand{\nbar}{\overline{n}}
\newcommand{\ibar}{\overline{i}}
\newcommand{\alphabar}{\overline{\alpha}}
\newcommand{\betabar}{\overline{\beta}}

\newcommand{\qK}{\boldsymbol{\mathrm{A}}} 
\newcommand{\qA}{\qK_{2n}}
\newcommand{\Rep}{\mathrm{Rep}} 

\newcommand{\injection}{\rightarrowtail}

\newcommand{\surjection}{\twoheadrightarrow}

\newcommand{\embedding}{\injection}

\newcommand{\modM}{\mathsf{M}}
\newcommand{\modN}{\mathsf{N}}
\newcommand{\modNbar}{\text{\mifody\sf I}}
\newcommand{\modW}{\mathsf{W}}

\newcommand{\soc}{\mathrm{soc}}
\newcommand{\rad}{\mathrm{rad}}

\renewcommand{\hat}{\widehat}

\newcommand{\idem}{\boldsymbol{e}}

\newcommand{\bref}[1]{\textbf{\ref{#1}}}

\newcommand{\im}{\mathop{\mathrm{im}}\nolimits}
\renewcommand{\ker}{\mathop{\mathrm{ker}}\nolimits}

\renewcommand{\geq}{\,{\geqslant}\,}
\renewcommand{\leq}{\,{\leqslant}\,}
\renewcommand{\ge}{\,{\geqslant}\,}
\renewcommand{\le}{\,{\leqslant}\,}

\newcommand{\FunF}{\mathcal{F}} 
\newcommand{\FunG}{\mathcal{G}} 

\newcommand{\HomC}{\mathrm{Hom}_{\mathbb{C}}}
\newcommand{\justHom}{\mathrm{Hom}}
\newcommand{\Hom}{\mathrm{Hom}_{\rule{0pt}{6.5pt}%
\mathscr{LU}_q}}
\newcommand{\Homm}{\mathrm{Hom}}
\newcommand{\Ext}{\mathrm{Ext}_{\rule{0pt}{6.5pt}%
    \mathscr{LU}_q}^1}

\newcommand{\ExtnU}{\mathrm{Ext}_{\rule{0pt}{6.5pt}%
    \overline{\mathscr{U}}_q}^n}
\newcommand{\ExtUbul}{\mathrm{Ext}_{\rule{0pt}{6.5pt}%
    \overline{\mathscr{U}}_q}^{\bullet}}
\newcommand{\Extn}{\mathrm{Ext}_{\rule{0pt}{6.5pt}%
    \mathscr{LU}_q}^n}
\newcommand{\EXT}{\mathrm{Ext}^\bullet}

\newcommand{\bP}{\mathbf{P}}
\newcommand{\morX}{X}
\newcommand{\modI}{\mathsf{Ex}}
\newcommand{\ob}{\mathrm{Ob}}

\newcommand{\tensor}{\otimes}

\newcommand{\q}{\mathfrak{q}}

\newcommand{\catC}{\mathscr{C}}
\newcommand{\lc}{\catC}

\newcommand{\catCpl}{\mathscr{C}^{+}}
\newcommand{\catCmin}{\mathscr{C}^{-}}
\newcommand{\catSpl}{\mathscr{S}^{+}}
\newcommand{\catSmin}{\mathscr{S}^{-}}

\newcommand{\fusion}{%
  \mathop{{\otimes}\kern-7pt\raisebox{.6pt}{%
      \mbox{\footnotesize ${\bullet}$}}}}

\newcommand{\UresSL}[1]{\overline{\mathscr{U}}_{\q} s\ell(#1)}
\newcommand{\LUresSL}[1]{\mathscr{LU}_{\q} s\ell(#1)}
\newcommand{\LUq}{\mathscr{LU}_{\q}}

\newcommand{\ffrac}[2]{\mbox{\footnotesize $\displaystyle\frac{#1}{#2}$}}
\newcommand{\half}{%
  \mathchoice{\ffrac{1}{2}}{\frac{1}{2}}{\frac{1}{2}}{\frac{1}{2}}}

\newcommand{\LM}{{\mathcal{LM}}}

\newcommand{\repX}{\mathsf{X}}

\newcommand{\modPP}{\mathbb{P}}
\newcommand{\modII}{\mathbb{I}}

\newcommand{\XX}{\mathsf{X}} 

\newcommand{\PP}{\mathsf{P}}
\newcommand{\modPr}{\PP}

\newcommand{\repF}{\rep{F}}

\newcommand{\rep}{\mathscr}  
\newcommand{\voal}{\mathcal} 
\newcommand{\Vir}{\voal V_{p}}

\newcommand{\modR}{\mathcal{R}}

\newcommand{\oC}{\mathbb{C}}

\newcommand{\oN}{\mathbb{N}}

\newcommand{\oZ}{\mathbb{Z}}

\newcommand{\one}{\boldsymbol{1}}

\newcommand{\toppr}{\mathsf{t}}
\newcommand{\botpr}{\mathsf{b}}
\newcommand{\leftpr}{\mathsf{l}}
\newcommand{\rightpr}{\mathsf{r}}
\newcommand{\midpr}{\mathsf{m}}

\newcommand{\stprp}{\mathsf{a}}

\newcommand{\cas}{\boldsymbol{C}}

\numberwithin{equation}{section}
\makeatletter
\@addtoreset{equation}{section}
\@addtoreset{subsubsection}{section}

\def\@secnumfont{\bfseries}
\def\subsubsection{\@startsection{subsubsection}{3}%
  \z@{.5\linespacing\@plus.7\linespacing}{-.5em}%
  {\normalfont\bfseries}}
\def\paragraph{\@startsection{paragraph}{4}%
  \z@\z@{-\fontdimen2\font}%
  \normalfont\bfseries}
\def\subparagraph{\@startsection{subparagraph}{5}%
  \z@\z@{-\fontdimen2\font}%
  \normalfont\bfseries}

\makeatother

\newcommand{\rme}{{\rm e}}

\newcommand{\step}{\mathrm{step}}

\swapnumbers
\newtheorem{Thm}[subsection]{Theorem}
\newtheorem{thm}[subsubsection]{Theorem}

\newtheorem{lemma}[subsubsection]{Lemma}
\newtheorem{Prop}[subsection]{Proposition}
\newtheorem{prop}[subsubsection]{Proposition}

\newtheorem{cor}[subsubsection]{Corollary}
\newtheorem{Conj}[subsection]{Conjecture}

\theoremstyle{definition}

\newtheorem{rem}[subsubsection]{Remark}

\begin{document}
\title[Kazhdan-Lusztig equivalence and fusion of Kac modules]{%
   Kazhdan-Lusztig equivalence and fusion
 of Kac modules in Virasoro 
   logarithmic models}
\author{P.V.~Bushlanov, A.M.~Gainutdinov and I.Yu.~Tipunin}
\address{PVB:Moscow Institute of Physics and Technology, Dolgoprudny, Institutskiy per. 9 , Russia, 141700 }
\email{paulbush@mail.ru} 
\address{AMG:Institut de Physique Th\'eorique, CEA Saclay, Gif Sur Yvette, 91191, France}
\email{azat.gaynutdinov@cea.fr} 
\address{IYuT:Tamm Theory Division, Lebedev Physics Institute, Leninski pr., 53,
Moscow, Russia, 119991}
\email{tipunin@gmail.com} 

\maketitle

\begin{abstract}
The subject of our study is the Kazhdan--Lusztig (KL) equivalence in
the context of a one-parameter family of logarithmic CFTs based on
Virasoro symmetry with the $(1,p)$ central charge.  All
finite-dimensional indecomposable modules of the KL-dual quantum group
-- the ``full'' Lusztig quantum $s\ell(2)$ at the root of unity -- are
explicitly described.  These are exhausted by projective modules and
four series of modules that have a functorial correspondence with any
quotient or a submodule of Feigin--Fuchs modules over the Virasoro
algebra.  Our main result includes calculation of tensor products of
any pair of the indecomposable modules.  Based on the Kazhdan-Lusztig
equivalence between quantum groups and vertex-operator algebras,
fusion rules of Kac modules over the Virasoro algebra in the $(1,p)$
LCFT models are conjectured.
\end{abstract}

\section{Introduction}
Logarithmic conformal field theories (LCFTs) have proven to be one of the
richest subjects in theoretical and mathematical physics with
applications in a ``non-unitary'' world ranging from modeling avalanche
processes~\cite{[PGPR]}, observables in stochastic processes
$\text{SLE}(\kappa,\rho)$~\cite{Kyt}, and surface critical behaviour in
O$(n)$ models and loop models~\cite{[DJS]}, to
percolation probabilities~\cite{[MRperc],[Ridout],[PRZ]}, and edge
states in the quantum Hall effect~\cite{ReadSaleur01,[BJS]}.
Beside the physical applications, LCFTs give rise to subjects of intense
studies from a more formal point of view including a free-field
representation~\cite{[GaberdielKausch3],[GK2],[FHST],[FGST3]},
vertex-operator algebras approach~\cite{[HLZ],[AM1]}, Zhu algebras
aspects and super W-algebras~\cite{[AM]}, quantum-group
dualities~\cite{[FGST],[FGST4],Semikh-q} and Verlinde
algebras~\cite{[FHST],[GT],[GR2],[PRR]}, construction of a new class of
W-algebras extending symmetry in rational CFTs based on affine
$s\ell(2)$~\cite{Semikh}, an interplay between rational boundary LCFTs and
non-semisimple braided finite tensor categories~\cite{[GabRW],[FSch]}, and recently in
defining a wide family of LCFTs parametrized by Dynkin
diagrams~\cite{[FT]}.

One of important achievements made in studying LCFTs has been the
systematic definition of chiral algebras in terms of so-called
screening currents~\cite{[FHST],[FGST3]}.  The idea is to use
screening operators intertwining a Virasoro-module structure on a
lattice VOA, and to define chiral algebras relevant for LCFTs as the
kernel of the screening operators. By contrast, chiral algebras
defining RCFTs are usually defined as cohomologies of the screening
operators~\cite{[F]}.  Subsequently, the ``screening kernel'' approach
in defining LCFTs has led to explicit construction of quantum groups
(at roots of unity) centralizing the chiral
algebras~\cite{[FGST3],[BFGT]}. Such quantum-group symmetry in the
space of states allowed to describe representation categories of the
chiral algebras~\cite{[FGST],[FGST4],[BFGT]} -- a
correspondence between  subquotient structures of modules 
 over the chiral algebra and over its centralizing quantum group, and
 between their fusion rules was stated\footnote{In the context of
 $(p,p')$ models, a  one-to-one correspondence, i.e. an
 equivalence, was stated only in the case $p'=1$ while for coprime $p,p'\ne1$ the
 correspondence was stated up to minimal models contribution.}. In a
 simplest LCFT such as symplectic fermions~\cite{[Kausch]} an
 equivalence on the level of braided tensor categories and modular
 group action was proven in~\cite{[FGST2]}.
Such nice correspondences, in a representation-theoretic (or functorial)
sense, give an extended analogue of the so-called Kazhdan--Lusztig
equivalence established first in a context of affine Lie algebras at
negative integer levels~\cite{[KL]}.

The Kazhdan--Lusztig (KL) equivalence was well tested so far in
several cases of rational LCFTs~\cite{[FGST],[FGST4]}, with a recent
contribution in~\cite{[ST]}. In particular,
a KL equivalence between a restricted (or ``small'') quantum $s\ell(2)$ at the
primitive $2p$th root of unity (based on a ``short'' screening and
denoted by $\UresSL2$) and the triplet $W$-algebra~\cite{[K-first]}
realizing an extended conformal symmetry in $(1,p)$ logarithmic
conformal field models was established in~\cite{[FGST]} while a proof
of this equivalence on the level of abelian representation categories
was given quite recently~\cite{Tsuchiya}.

Another achievement in defining and studying logarithmic theories was
proposed few years ago~\cite{RS2,RS-Q,[PRZ]}, based on a
construction of lattice discretizations of the LCFTs.  The point is
that it is still difficult to compute fusion rules and determine
subquotient structure of indecomposable representations of the chiral
algebra that should appear in continuum logarithmic theories.  Lattice
discretizations, on the other hand, naturally involve well-studied
``lattice'' algebras such as Temperley--Lieb
(TL)~\cite{[TL],[Jones],[GL]}, Brauer~\cite{[Br],[Mart2]}, different
types of blob algebras~\cite{[MarSal],[GL1],[DJS]} and their
centralizing quantum groups~\cite{Mart1,[Good]} as well.  The transfer
matrix and the Hamiltonian operators are particular elements of these
``lattice'' algebras and much intuition as well as rigorous results
can be obtained from the study of these lattice features. In
particular, the blob (or boundary TL) algebras give a quick access to
a description of (integrable) boundary conditions which carry over
rather straightforwardly to Virasoro-symmetric boundary conditions in
the continuum limit~\cite{[PRZ],RS-Q,[DJS]} giving thus examples of
non-rational  LCFTs which involve infinitely many primary
fields and their logarithmic partners.

The purpose of this paper is to accomplish an important step forward
in the study of non-rational LCFTs by using the KL equivalence. The
subject of our study is a one-parameter family of chiral logarithmic
models with the $(1,p)$ central charge $c_{1,p}=13-6/p-6p$, 
 with integer
$p\geq2$. These models were originally formulated as the continuum
limit of XXZ spin-chains at appropriate roots of unity~\cite{RS-Q} and
as the limit of integrable lattice face models~\cite{[PRZ]}. Both are based on the TL
algebra which morally gives a regularization of the stress-energy
tensor modes (Virasoro generators) on a finite
system~\cite{KooSaleur}. The chiral algebra for these models in the
continuum is (a representation of) the Virasoro algebra of the
central charge $c_{1,p}$. We will denote these LCFTs as~$\LM(1,p)$, following
notations in~\cite{[PRZ]}. Fusion rules in these conformal models
were originally studied using an implementation of the
Gaberdiel--Kausch--Nahm (GKN) algorithm~\cite{GKfus} and then
investigated in a more systematic way in~\cite{RS-Q} and
in~\cite{RPfus,JR}, combining lattice computations with the GKN
algorithm.
 
  We conjectured in~\cite{[BFGT]} that a KL-type equivalence exists
between a ``long'' screening extension of $\UresSL2$ -- the Lusztig
limit $\LUresSL2$ of the full quantum $s\ell(2)$ as~$\q\to
e^{\imath\pi/p}$ -- and the Virasoro vertex-operator algebra $\Vir$ defined
by the $SL(2)$-invariant subspace in the vacuum module of the triplet
$W$-algebra.  Moreover, the fact that  the TL
algebra and a representation of the quantum group $\LUresSL2$ centralize each
other in the XXZ spin chains~\cite{PasqS,Mart1,RS-Q} suggests that in the continuum limit
many results about $\LM(1,p)$ can be reformulated in terms
of~$\LUresSL2$ too.  For existence of the KL equivalence in the
context of LCFTs with the Virasoro chiral algebra, there are pieces of evidence
which we bring in series.

In~\cite{[BFGT]}, the tensor products of all irreducible and projective
modules over $\LUresSL2$ were calculated and identified with the fusion of 
irreducible and logarithmic (staggered) modules over $\Vir$. In this
paper, we give an exhaustive description of all indecomposable modules
in the category of finite dimensional $\LUresSL2$-modules and
calculate all their tensor products.  This allows us to identify
indecomposable $\LUresSL2$-modules with indecomposable Kac modules\footnote{By a
Kac module associated with any pair of integers $(r,s)$, $r,s\geq1$,
we call the quotient of the corresponding Feigin--Fuchs module~\cite{FF} by a
singular vector on the level $rs$, see a precise definition below.}
over $\Vir$, by comparison between their subquotient structure, i.e. on
the level of abelian categories, using our results
from~\cite{[BFGT]}. Then, we conjecture fusion rules for
the Kac modules  using tensor product
decompositions for the corresponding  
modules over $\LUresSL2$. Our results agree with ones in~\cite{JR}.

\medskip

\subsection{Results}\label{sec:results}
In order to describe results of this paper, we use quite basic terminology
in category theory~\cite{[Kassel]}. Our results consist in the following.
Let  $\catC_p$ denotes the category of finite dimensional
$\LUresSL2$-modules. Then, $\catC_p$
 is a direct sum of two full subcategories
$\catC_p=\catC^+_p\oplus\catC^-_p$ such that
there are no morphisms between $\catC_p^+$ and $\catC_p^-$ and the
subcategory $\catC_p^+$ is closed under tensor products. 
Categories $\catC^+_p$ and $\catC^-_p$
are equivalent as abelian categories. We let $\nu:\catCpl_p\to\catCmin_p$ denotes
an equivalence functor.
The tensor product of any two objects in $\catC^-_p$ belongs 
to $\catC^+_p$ and tensor product of an object from 
$\catC^+_p$ with an object from $\catC^-_p$ belongs to $\catC^-_p$.
This determines a $\oZ_2$ structure on the tensor
category $\catC_p$. To calculate tensor product for any pair of objects in $\catC_p$,
it is enough to know tensor products in $\catC^+_p$.
Let $Y$ be an object from $\catC^+_p$ and $Y'$, $Y''$ are objects from $\catC^-_p$.
Then, $Y\otimes Y'=\nu(Y\otimes\nu^{-1}(Y'))$ and
$Y'\otimes Y''=\nu^{-1}(Y')\otimes\nu^{-1}(Y'')$.
Therefore, we describe only structure of~$\catC^+_p$ in
detail.

The set of indecomposable modules in the category $\catC^+_p$ consists of
irreducible modules $\XX_{s,r}$, for any pair of integers $1\leq s\leq p$ and
$r\geq1$, their projective covers $\PP_{s,r}$
(which are simultaneously projective and injective objects in~$\catC^+_p$)  and 
 modules $\modN_{s,r}(n)$, $\modNbar_{s,r}(n)$,
$\modM_{s,r}(n)$ and $\modW_{s,r}(n)$, with $1\leq s\leq p-1$ and
$r,n\geq1$. To describe briefly 
the irreducible module $\XX_{s,r}$, we note that it is a tensor product 
of $s$-dimensional irreducible $\UresSL2$-module and $r$-dimensional irreducible
$s\ell(2)$-module (see precise definitions in Sec.~\bref{sec:irreps}).
The projective cover $\PP_{s,r}$ of $\XX_{s,r}$
 has the following subquotient structure:
\begin{equation}\label{schem-proj}
 \xymatrix@=12pt{
    &\stackrel{\XX_{s,1}}{\bullet}\ar[d]_{}
    \\
    &\stackrel{\XX_{p - s, 2}}{\circ}\ar[d]^{}
    \\
    &\stackrel{\XX_{s,1}}{\bullet}
  }\qquad\qquad\qquad
  \xymatrix@=12pt{
    &&\stackrel{\XX_{s,r}}{\bullet}
    \ar@/^/[dl]
    \ar@/_/[dr]
    &\\
    &\stackrel{\XX_{p - s, r - 1}}{\circ}\ar@/^/[dr]
    &
    &\stackrel{\XX_{p - s, r + 1}}{\circ}\ar@/_/[dl]
    \\
    &&\stackrel{\XX_{s,r}}{\bullet}&
  }
\end{equation}
where $r\geq2$ and $1\leq s\leq p-1$.
We note also that $\PP_{p,r}=\XX_{p,r}$.
The set of irreducible and
projective modules is closed under tensor products.

All other indecomposable objects in $\catC^+_p$
 have subquotient
structures in  the following list (see also Thm~\bref{thm:cat-decomp}), where we set $1\leq s\leq p-1$ and
$r,n\geq1$.
\begin{itemize}
 \item[$\smash{{\modW_{s,r}(n)}}$:]
\begin{equation}\label{schem-W}
   \xymatrix@C=10pt@R=12pt{
     \stackrel{\repX_{s,r}}{\circ}\ar@/^/[dr]&
     &\stackrel{\repX_{s,r+2}}{\circ}\ar@/_/[dl]
     \ar@/^/[dr]&&\dots\ar@/_/[dl]\ar@/^/[dr]&
     &\stackrel{\repX_{s,r+2n}}{\circ}
     \ar@/_/[dl]\\
     &\stackrel{\quad\repX_{p-s,r+1}\quad}{\bullet}&&
     \stackrel{\quad\repX_{p-s,r+3}\quad\,}{\bullet}&\dots&
     \stackrel{\quad\repX_{p-s,r+2n-1}\quad}{\bullet}&
   }
 \end{equation}
\item[${\modM_{s,r}(n)}$:]
 \begin{equation}\label{schem-M}
      \xymatrix@C=10pt@R=12pt{
     &\stackrel{\repX_{p-s,r+1}}{\circ}\ar@/_/[dl]
     \ar@/^/[dr]
     &&\stackrel{\repX_{p-s,r+3}}{\circ}\ar@/_/[dl]
     \ar@/^/[dr]
     &&\dots\ar@/_/[dl]\ar@/^/[dr]&
     &\stackrel{\repX_{p-s,r+2n-1}}{\circ}
     \ar@/_/[dl]\ar@/^/[dr]&\\
     \stackrel{\repX_{s,r}}{\bullet}&&
     \stackrel{\repX_{s,r+2}}{\bullet}
     &&\stackrel{\repX_{s,r+4}}{\bullet}
     &\dots&
     \stackrel{\repX_{s,r+2n-2}}{\bullet}&&
     \stackrel{\repX_{s,r+2n}}{\bullet}
   }\kern-40pt
 \end{equation}
\item[${\modN_{s,r}(n)}$:]
\begin{equation}\label{schem-N}
   \xymatrix@C=10pt@R=12pt{
     &\stackrel{\repX_{p-s,r+1}}{\circ}\ar@/_/[dl]
     \ar@/^/[dr]
     &&\stackrel{\repX_{p-s,r+3}}{\circ}\ar@/_/[dl]
     \ar@/^/[dr]
     &&\dots\ar@/_/[dl]\ar@/^/[dr]&
     &\stackrel{\repX_{p-s,r+2n-1}}{\circ}
     \ar@/_/[dl]&\\
     \stackrel{\quad\repX_{s,r}\;}{\quad\bullet}&&
     \stackrel{\;\repX_{s,r+2}\;}{\bullet}
     &&\stackrel{\;\repX_{s,r+4}\;}{\bullet}
     &\dots&
     \stackrel{\;\repX_{s,r+2n-2}\;}{\bullet}&&
   }\kern-40pt
 \end{equation} 
\item[$\smash{{\modNbar_{s,r}(n)}}$:]
 \begin{equation}\label{schem-Nbar}
\xymatrix@C=10pt@R=12pt{
     &\stackrel{\repX_{s,r}}{\circ}\ar@/^/[dr]&
     &\stackrel{\repX_{s,r+2}}{\circ}\ar@/_/[dl]
     \ar@/^/[dr]&&\dots\ar@/_/[dl]\ar@/^/[dr]&
     &\stackrel{\repX_{s,r+2n-2}}{\circ}
     \ar@/_/[dl]\ar@/^/[dr]&&\\
     &&\stackrel{\repX_{p-s,r+1}}{\bullet}&&
     \stackrel{\repX_{p-s,r+3}}{\bullet}&\dots&
     \stackrel{\repX_{p-s,r+2n-3}}{\bullet}&&
     \stackrel{\repX_{p-s,r+2n-1}}{\bullet}&
   }\kern-40pt
 \end{equation}
\end{itemize}

\bigskip

In~\cite{[BFGT]}, it was conjectured that the category $\catC_p^+$ is equivalent as 
a tensor category to the
category of Virasoro algebra representations appearing in
$\LM(1,p)$. Under the equivalence, irreducible and projective modules
are identified in the following way
\begin{equation}\label{identification}
\begin{split}
\repX_{p,2r-1}\to \modR^0_{2r-1},\quad
\repX_{p,2r}\to \modR^0_{2r},\quad
\PP_{s,2r-1}\to \modR^{p-s}_{2r-1},\quad
\PP_{p-s,2r}\to \modR^{s}_{2r},\\
\repX_{s,2r-1}\to(2r-1,s),\qquad
\repX_{s,2r}\to(2r,s), \qquad 1\leq s\leq p,\quad r\geq 1,
\end{split}
\end{equation}
where $(r,s)$ are the irreducible Virasoro modules with the highest
weights 
\begin{equation}\label{Delta}
 \Delta_{r,s} = ((pr-s)^2-(p-1)^2)/4p
\end{equation} 
and the $\modR^{s}_{r}$
are logarithmic  Virasoro modules\footnote{We note that in
order to define $\modR^{s}_{r}$, for $r>1$, as Virasoro modules it is necessary to
say about the value of so-called $\beta$-invariant~\cite{[KytRid]} for
these modules. To determine these numbers is out of the scope of the paper.} (also known as staggered
modules~\cite{[KytRid]}), in notations of~\cite{RPfus}.  Under this
identification, the fusion of the irreducible and staggered $\Vir$-modules is given by
the tensor products of the corresponding $\LUresSL2$-modules.

Let $\repF_{s,m}$, with $1\leq s\leq p$ and $m\in\oZ$, be the Feigin-Fuchs module~\cite{FF} over
$\Vir$ with the lowest conformal dimension $\Delta_{m,p-s} =
\Delta_{1-m,s}$ (see precise definitions in App.~\bref{feigin_fuchs}).
Using the identifications for irreducible modules
in~\eqref{identification}, the indecomposable modules~\eqref{schem-W}-\eqref{schem-Nbar} over
$\LUresSL2$ are then identified with quotients and submodules of the
Feigin-Fuchs modules. The two families of the $\modW$-type
modules from~\eqref{schem-W} and of the $\modN$-type
from~\eqref{schem-N} have the correspondence
\begin{equation}\label{Kac-mod-id}
\begin{split}
\modW_{s,r}(n)\to \repF_{p-s,r}/\repF_{p-s,-r-2n},\\
\modN_{s,r}(n)\to \repF_{s,1-r}/\repF_{s,1-r-2n}.
\end{split}
\end{equation}
Modules from the other two families, the $\modM$-type modules from~\eqref{schem-M} and
the $\modNbar$-type from~\eqref{schem-Nbar}, are contragredient
to the $\modW$- and $\modN$-type modules, respectively. Equivalently, they are identified with submodules in Feigin-Fuchs modules:
each $\modNbar_{s,r}(n)$ consists of $2n$ subquotients of
$\repF_{p-s,r}$ in Fig.~\ref{felder-complex}, see App.~\bref{feigin_fuchs}, that appear
from left to right; each $\modM_{s,r}(n)$ consists of $2n+1$
subquotients of $\repF_{s,1-r}$. We call  these quotients and
submodules of the  Feigin-Fuchs modules by Kac modules.

Having the identification between subquotient structures of
indecomposable qunatum-group modules and Kac modules, it is
interesting to note that dimensions of subquotients over $\LUresSL2$
count conformal levels of corresponding (sub)singular vectors in the
Kac modules. Recall that a singular vector satisfies the
highest-weight conditions and it thus belongs to the left-most node or
to any node of the type `$\bullet$' in
our diagrams  while a subsingular vector
satisfies the highest-weight conditions only in a quotient by one of
the $\bullet$-submodules, i.e. it belongs to a subquotient labeled by
`$\circ$' in our diagrams for modules.  Then, the conformal level of a
(sub)singular vector in a Kac module is given by the sum of dimensions of all
irreducible quantum-group subquotients that are on the left from the
subquotient of the (sub)singular vector in the corresponding diagram
for $\LUresSL2$.

\begin{rem}
On the category $\catC^+_p$, there is a functor ${\cdot}^*$
which maps each object to its contragredient one, with all arrows reversed.
In particular, it acts on indecomposable modules as
$\repX_{s,r}^*=\repX_{s,r}$, $\PP_{s,r}^*=\PP_{s,r}$, 
$\modN_{s,r}(n)^*=\modNbar_{s,r}(n)$, $\modNbar_{s,r}(n)^*=\modN_{s,r}(n)$, 
$\modM_{s,r}(n)^*=\modW_{s,r}(n)$ and $\modW_{s,r}(n)^*=\modM_{s,r}(n)$.
 In addition, the functor ${\cdot}^*$ is a tensor functor
\begin{equation}
 (X\tensor Y)^*=X^*\tensor Y^*.
\end{equation}
\end{rem}

\medskip 

We now describe tensor product decompositions for all indecomposable
modules over $\LUresSL2$.
Formulas for tensor products of indecomposble modules in $\catC^+_p$
are quite cumbersome and to write them we introduce the following notation
\begin{gather}
\gamma_1=(s_1+s_2+1)\!\!\!\!\mod2,\qquad \gamma_2=(s_1+s_2+p+1)\!\!\!\!\mod2,\notag\\
\mathop{\bigoplus{\kern-3pt}'}\limits_{r=a}^{b}f(r) = 
\bigoplus_{r=a}^{b}(1 - \ffrac{1}{2}\delta_{r,a} - \ffrac{1}{2}\delta_{r,b})f(r),\label{not-1}\\
\sg(r)=\begin{cases}\phantom{-}1,\quad r>0,\\ 
      \phantom{-}0,\quad r=0,\\-1,\quad r<0.\end{cases}\notag
\end{gather}
We do not write all possible tensor products because there are simple 
rules, which use commutativity
and associativity of the tensor product, giving all tensor products
from base ones -- the tensor products of simplest indecomposables,
like $\modN_{s,r}(1)$,  with
irreducible modules and of the simplest indecomposables with themselves.
The base tensor products are collected in the following theorem.
\begin{Thm}\label{thm:tens-prod-intro}
\begin{enumerate}
 \item The tensor product of two irreducible modules with $s_1,s_2=1,\dots, p$
and $r\in\oN$ is
\begin{equation*}
\repX_{s_1,r_1}\tensor\repX_{s_2,r_2} =
\bigoplus_{\substack{r=|r_1-r_2|+1\\\step=2}}^{r_1+r_2-1}
\Bigr(\bigoplus_{\substack{s=|s_1-s_2|+1\\\step=2}}^{\substack{
\min(s_1 + s_2 - 1,\\ 2p - s_1 - s_2 - 1)}}\!\!\!\repX_{s,r}
\;\oplus\!\!\! \bigoplus_{\substack{s=2p - s_1 - s_2 +1\\\step=2}}^{p-\gamma_2}
\!\!\!\!\!\!\PP_{s,r}\Bigl)
\end{equation*}
\item The tensor product of an irreducible with a projective module
with $s_1=1,\dots, p$, $s_2=1,\dots, p-1$
and $r\in\oN$ is
\begin{equation*}
\repX_{s_1,r_1}\tensor\PP_{s_2,r_2} = 
\bigoplus_{\substack{r=|r_1-r_2|+1\\\step=2}}^{r_1+r_2-1}\Bigl(
\bigoplus_{\substack{s=|s_1-s_2|+1\\\step=2}}^{\substack{
\min(s_1 + s_2 - 1,\\ 2p - s_1 - s_2 - 1)}}\!\!\!\!\!\!\PP_{s,r}
\oplus 2\!\!\!\!\!\!\bigoplus_{\substack{s=2p-s_1-s_2+1\\\step=2}}^{p-\gamma_2}
\!\!\!\!\!\!\PP_{s,r}\Bigr)
\oplus 2\mathop{\bigoplus{\kern-3pt}'}\limits_{\substack{r=|r_1-r_2|\\\step=2}}^{r_1+r_2}
\bigoplus_{\substack{s=p-s_1+s_2+1\\\step=2}}^{p-\gamma_1}
\!\!\!\!\!\!\PP_{s,r},
\end{equation*}
where we set $\PP_{s,0}=0$.
\item The tensor products of an irreducible module with simplest $\modN$-type
  modules with
$s_1,s_2=1,\dots, p-1$
and $r_1,r_2\in\oN$ are
\begin{equation*}
\begin{split}
&\repX_{s_1,r_1}\tensor\modN_{s_2,r_2}(1) = \bigoplus_{\substack{r=|r_1-r_2|+1\\\step=2}}^{r_1+r_2-1}  
\quad \bigoplus_{\substack{s=2p - s_1 - s_2 + 1\\\step=2}}^{p-\gamma_2}\!\!\!\!\!\!\PP_{s,r}
\oplus \bigoplus_{\substack{r=|r_1-r_2-1|+1\\\step=2}}^{r_1+r_2} \quad 
\bigoplus_{\substack{s= p + s_2 - s_1 + 1\\\step=2}}^{p-\gamma_1}\!\!\!\!\!\!\PP_{s,r}\,\oplus\\
&\kern200pt\oplus\bigoplus_{\substack{s=|s_1-s_2|+1\\\step=2}}^{\substack{\min(s_1 + s_2 - 1,\\ 2p - s_1 - s_2 - 1)}}\!\!\!\begin{cases}
&\modN_{s,r_2-r_1+1}(r_1),\quad r_1 \le r_2\\
&\modW_{p-s,r_1-r_2}(r_2),\quad r_1 > r_2
\end{cases}
\end{split}
\end{equation*}
\item 
 The tensor products of an irreducible module with simplest $\modNbar$-type
  modules with
$s_1,s_2=1,\dots, p-1$
and $r_1,r_2\in\oN$ are
\begin{equation*}
\begin{split}
&\repX_{s_1,r_1}\tensor\modNbar_{s_2,r_2}(1) = 
\bigoplus_{\substack{r=|r_1-r_2|+1\\\step=2}}^{r_1+r_2-1}  \quad 
\bigoplus_{\substack{s=2p - s_1 - s_2 + 1\\\step=2}}^{p-\gamma_2}\!\!\!\!\!\!\PP_{s,r}
\oplus\bigoplus_{\substack{r=|r_1-r_2-1|+1\\\step=2}}^{r_1+r_2} \quad 
\bigoplus_{\substack{s= p + s_2 - s_1 + 1\\\step=2}}^{p-\gamma_1}\!\!\!\!\!\!\PP_{s,r}\,\oplus\\
&\kern200pt\oplus\bigoplus_{\substack{s=|s_1-s_2|+1\\\step=2}}^{\substack{\min(s_1 + s_2 - 1,\\ 2p - s_1 - s_2 - 1)}}\!\!\!\begin{cases}
&\modNbar_{s,r_2-r_1+1}(r_1),\quad r_1 \le r_2\\
&\modM_{p-s,r_1-r_2}(r_2),\quad r_1 > r_2
\end{cases}
\end{split}
\end{equation*}
\item The tensor products of two simplest $\modN$-type modules
with $s_1,s_2=1,\dots, p-1$
and $r_1,r_2\in\oN$ are
\begin{equation*}
\begin{split}
&\modN_{s_1,r_1}(1)\tensor\modN_{s_2,r_2}(1) = \bigoplus_{\substack{r=|r_1-r_2|+1\\\step=2}}^{r_1+r_2-1}
\bigoplus_{\substack{s=|s_1-s_2|+1\\\step=2}}^{p-\gamma_2}\!\!\!\PP_{s,r}
\oplus \bigoplus_{\substack{s=s_1 + s_2 + 1\\\step=2}}^{p-\gamma_2}\!\!\!\!\!\!\PP_{s,r_1+r_2+1}\,\oplus\\
&\oplus\bigoplus_{\substack{r=|r_1-r_2+\sg(s_2-s_1)|+1\\\step=2}}^{r_1+r_2}
\bigoplus_{\substack{s=p-|s_1-s_2|+1\\\step=2}}^{p-\gamma_1}\!\!\!\PP_{s,r} 
\oplus \bigoplus_{\substack{s=|p-s_1-s_2|+1\\\step=2}}^{p - |s_1 - s_2| - 1}\!\!\!\modN_{s,r_1+r_2}(1)
\end{split}
\end{equation*}
\item  The tensor products of two simplest $\modNbar$-type modules
with $s_1,s_2=1,\dots, p-1$
and $r_1,r_2\in\oN$ are
\begin{align*}
&\modNbar_{s_1,r_1}(1)\tensor\modNbar_{s_2,r_2}(1) = 
\bigoplus_{\substack{r=|r_1-r_2|+1\\\step=2}}^{r_1+r_2-1} \bigoplus_{\substack{s=|s_1-s_2|+1\\\step=2}}^{p-\gamma_2}\!\!\!\PP_{s,r}\;
\oplus\!\!\! \bigoplus_{\substack{s=s_1 + s_2 +1\\\step=2}}^{p-\gamma_2}\!\!\!\!\!\!\PP_{s,r_1+r_2+1}\,\oplus\\
&\oplus\bigoplus_{\substack{r=|r_1-r_2+\sg(s_2-s_1)|+1\\\step=2}}^{r_1+r_2}
\bigoplus_{\substack{s=p-|s_1-s_2|+1\\\step=2}}^{p-\gamma_1}\!\!\!\PP_{s,r}
\oplus\bigoplus_{\substack{s=|p-s_1-s_2|+1\\\step=2}}^{p - |s_1 - s_2| - 1}\!\!\!\modNbar_{s,r_1+r_2}(1)
\end{align*}
\item The tensors products of simplest $\modN$-type with simplest
  $\modNbar$-type modules with
$s_1,s_2=1,\dots, p-1$
and $r_1,r_2\in\oN$ are
\begin{align*}
&\modN_{s_1,r_1}(1)\tensor\modNbar_{s_2,r_2}(1) =\bigoplus_{\substack{r=|r_1-r_2|+2\\\step=2}}^{r_1+r_2}
\bigoplus_{\substack{s=|p-s_1-s_2|+1\\\step=2}}^{p-\gamma_1}\!\!\!\PP_{s,r}\,\oplus\\
&\oplus\!\!\! \bigoplus_{\substack{s=p-|s_1 - s_2| +1\\\step=2}}^{p-\gamma_1}\!\!\!\!\!\!
\delta_{\sg(r_1-r_2),\sg(s_1-s_2)}\PP_{s,|r_1-r_2|}
\oplus\kern-10pt\bigoplus_{\substack{r=|r_1-r_2|+1\\\step=2}}^{r_1+r_2+\sg(p-s_1-s_2)}\kern-10pt
\bigoplus_{\substack{s=\min(s_1 + s_2 + 1,\\ 2p - s_1 - s_2 + 1)\\\step=2}}^{p-\gamma_2}\!\!\!\PP_{s,r}\,\oplus\\
&\kern200pt\oplus\bigoplus_{\substack{s=|p-s_1-s_2|+1\\\step=2}}^{p-|s_1 - s_2| - 1}\!\!\!
\begin{cases}
&\modN_{s,r_1-r_2}(1),\quad r_1>r_2\\
&\modNbar_{s,r_2-r_1}(1),\quad r_2>r_1\\
&\repX_{p-s,1},\quad r_1=r_2.
\end{cases}
\end{align*}
\end{enumerate}
\end{Thm}
The tensor product of arbitrary two indecomposable modules can be obtained 
from the base tensor products given in the previous theorem and the
following list of rules, see also Thm.~\bref{thm:tensor-third}.
\begin{enumerate}
 \item The tensor product of $\PP_{s,r}$ with an indecomposable 
module is isomorphic to the tensor product of $\PP_{s,r}$ with the direct
sum of all irreducible subquotients constituting the  indecomposable 
module.

\item An arbitrary indecomposable module of the $\modW$-, $\modM$-,
  $\modN$- or $\modNbar$-type is isomorphic to the tensor product of an irreducible
module and a simplest indecomposable module:
\begin{align*}
\modN_{s,r}(n) &= \repX_{1,n}\tensor\modN_{s,r+n-1}(1),\\
\modNbar_{s,r}(n) &= \repX_{1,n}\tensor\modNbar_{s,r+n-1}(1),\\
\modW_{s,r}(n) &= \repX_{1,r+n}\tensor\modN_{p-s,n}(1),\\
\modM_{s,r}(n) &= \repX_{1,r+n}\tensor\modNbar_{p-s,n}(1),
\end{align*}
where $s=1,\dots, p-1$ and $r,n\in\oN$.
\end{enumerate}
Thm.~\bref{thm:tens-prod-intro} with  these rules completes the description of the tensor structure on
$\catC^+_p$. 

Following our previous result~\cite{[BFGT]} about the KL equivalence established for
a subcategory in $\catC^+_p$ containing all simple objects and their
projective covers, we now propose the following conjecture, which was
also mentioned in~\cite{[BFGT]}.
\begin{Conj}\label{conj-main}
The category $\catC^+_p$ is equivalent as a tensor category to the 
representation category of the vertex operator algebra $\Vir$ realized 
in $\LM(1,p)$.
\end{Conj}

Thus, we can conjecture fusion rules for the Kac modules over $\Vir$ using the
identification in~\eqref{Kac-mod-id} together with
Thm.~\bref{thm:tens-prod-intro} and the tensor-products rules (1) and
(2) described above.

The conjecture~\bref{conj-main} is also motivated by the fact that tensor product
decompositions of the indecomposable modules in $\catC^+_p$ coincide
with the fusion proposed in~\cite{JR} for the 
Kac modules from $\LM(1,p)$, using the identification 
 in notations of~\cite{JR}
\begin{equation*}
\begin{split}
\modNbar_{p - s,n - r + 1}(r)\to (r,s + np),\quad \mbox{whenever $2r - 1 < 2n$},\\
\modM_{s,r - n}(n)\to (r,s + np),\quad \mbox{whenever $2r - 1 > 2n$}
\end{split}
\end{equation*}
and
\begin{equation*}
\begin{split}
\modN_{p - s,n - r + 1}(r)\to (r,s + np)^*,\quad \mbox{whenever $2r - 1 < 2n$},\\
\modW_{s,r - n}(n)\to (r,s + np)^*,\quad \mbox{whenever $2r - 1 > 2n$}.
\end{split}
\end{equation*}

\medskip

The paper is organised as follows. In Sec.~\bref{sec:QG-VOA},
we recall a definition of the Hopf algebra $\LUresSL2$ by generators and relations 
and define their irreducible and projective modules.
In Sec.~\bref{sec:LU-rep}, we calculate $\EXT$'s between 
irreducible $\LUresSL2$-modules and obtain from this a classification
theorem of all indecomposable $\LUresSL2$-modules. In Sec.~\bref{sec:tens-prod},
we calculate decomposition of tensor products of all indecomposable
$\LUresSL2$-modules. Sec.~\bref{sec:concl} contains our conclusions.
Some technicalities and general well-known facts are arranged into six Appendices. App.~\bref{feigin_fuchs}
contains necessary information about Feigin--Fuchs modules over $\Vir$. App.~\bref{app:proj-mod-base}
and App.~\bref{app:indecomp-mod-base} contain explicit description of
indecomposable $\LUresSL2$-modules in terms of bases and action. App.~\bref{sec:resolutions}
contains our result about projective resolutions for irreducible
$\LUresSL2$-modules which are used in computation of  $\EXT$'s
groups. In   App.~\bref{app:prod_ind}, we give an exhaustive list  of
tensor products of indecomposable modules. 
 App.~\bref{app:quivers} contains necessary information about quivers
 which we use to prove the classification theorem.

\subsection{Notations}
In the paper, $\oN$ denotes the set of all integer $n\geq1$. We also set 
\begin{equation*}
  \q=e^{\frac{i\pi}{p}},
\end{equation*}
for any integer $p\geq2$,
 and use the standard notation
\begin{equation*}
  [n] = \ffrac{\q^n-\q^{-n}}{\q-\q^{-1}},\quad
  n\in\oZ,\quad [n]! = [1][2]\dots[n],\quad n\in\oN,\quad[0]!=1.
\end{equation*}

For Hopf algebras, we write
$\Delta$, $\epsilon$, and~$S$ for the comultiplication, counit, and
antipode respectively.

\section{Conventions and definitions.\label{sec:QG-VOA}}
In
setting the notation and recalling the basic facts about
$\mathscr{LU}_{\q}\equiv\LUresSL2$ needed below, we largely
follow~\cite{[BFGT]}. We collect the definitions of different quantum
groups in~\bref{sec:def-restr} and~\bref{sec:def-lusz},  and
recall basic facts about their representation theory
in~\bref{sec:irreps} and~\bref{proj-mod}.


\subsection{The restricted quantum group}\label{sec:def-restr}
The quantum group $\UresSL2$ is the ``restricted'' quantum $s\ell(2)$
with $\q = e^{i\pi/p}$ and the  generators $E$, $F$, and
$K^{\pm1}$ satisfying the standard relations for the quantum $s\ell(2)$,
\begin{equation}\label{Uq-com-relations}
  KEK^{-1}=\q^2E,\quad
  KFK^{-1}=\q^{-2}F,\quad
  [E,F]=\ffrac{K-K^{-1}}{\q-\q^{-1}},
\end{equation}
with the additional relations
\begin{equation}\label{root-1-rel}
  E^{p}=F^{p}=0,\quad K^{2p}=\one,
\end{equation}
and the Hopf-algebra structure is given by
\begin{gather}
  \Delta(E)=\one\otimes E+E\otimes K,\quad
  \Delta(F)=K^{-1}\otimes F+F\otimes\one,\quad
  \Delta(K)=K\otimes K,\label{Uq-comult-relations}\\
  S(E)=-EK^{-1},\quad  S(F)=-KF,\quad S(K)=K^{-1},
  \label{Uq-antipode}\\
  \epsilon(E)=\epsilon(F)=0,\quad\epsilon(K)=1.\label{Uq-epsilon}
\end{gather}


\subsubsection{Central idempotents}\label{app:idem}
We recall here a description of primitive central  idempotents in
$\UresSL2$ following~\cite{[FGST]}. Let $\cas$ denotes
 the Casimir element
\begin{equation}\label{eq:casimir}
  \cas=(\q-\q^{-1})^2 EF+\q^{-1}K+\q K^{-1}=(\q-\q^{-1})^2 FE+\q
  K+\q^{-1}K^{-1}.
\end{equation}
\newpage
The $\UresSL2$ has $p+1$ primitive central idempotents $\idem_s$,
$\sum_{s=0}^{p}\idem_s = \one$, which are the following polynomials in $\cas$:
\begin{align*}
  \idem_s&=\ffrac{1}{\psi_s(\beta_s)}\bigl(
  \psi_s(\cas)-\frac{\psi_s^\prime(\beta_s)}{\psi_s(\beta_s)}(\cas
   -\beta_s)\psi_s(\cas)\bigr),\quad
  1\leq s\leq p-1,\\
  \idem_0&=\ffrac{1}{\psi_0(\beta_0)}\psi_0(\cas),\qquad
  \idem_p=\ffrac{1}{\psi_p(\beta_p)}\psi_p(\cas),
\end{align*}
where  $\beta_j=\q^j+\q^{-j}$ and
 \begin{align*}
   \psi_s(x)&=(x-\beta_0)\,(x-\beta_p)
   \smash{\prod_{j=1, j\neq s}^{p-1}}(x-\beta_j)^2,\quad 1\leq s\leq p-1,\\
      \psi_0(x)&=(x-\beta_p)\prod_{j=1}^{p-1}(x-\beta_j)^2, \quad 
      \psi_p(x)=(x-\beta_0)\prod_{j=1}^{p-1}(x-\beta_j)^2.
\end{align*}

\subsection{The centralizer of $\Vir$}\label{sec:def-lusz} Here, we recall the quantum group 
$\LUresSL2$ (i.e. a Hopf algebra) that commutes with the Virasoro
algebra $\Vir$ on the chiral space of states~\cite{[BFGT]} associated
with the logarithmic Virasoro models $\LM(1,p)$.

\subsubsection{Definition} The Hopf-algebra
structure on $\LUresSL2$ is the following. The defining relations
between the $E$, $F$, and $K^{\pm1}$ generators are the same as in $\UresSL2$
and given in~\eqref{Uq-com-relations} and~\eqref{root-1-rel},
and the  $e$, $f$, and $h$ generators have  the usual $s\ell(2)$ relations
\begin{equation}\label{sl2-rel}
  [h,e]=e,\qquad[h,f]=-f,\qquad[e,f]=2h,
\end{equation}
while  ``mixed'' relations are
\begin{gather}
  [h,K]=0,\qquad[E,e]=0,\qquad[K,e]=0,\qquad[F,f]=0,\qquad[K,f]=0,\\
  [F,e]= \ffrac{1}{[p-1]!}K^p\ffrac{\q K-\q^{-1} K^{-1}}{\q-\q^{-1}}E^{p-1},\qquad
  [E,f]=\ffrac{(-1)^{p+1}}{[p-1]!} F^{p-1}\ffrac{\q K-\q^{-1} K^{-1}}{\q-\q^{-1}},
    \label{Ef-rel}\\
  [h,E]=\ffrac{1}{2}EA,\quad[h,F]=- \ffrac{1}{2}AF,\label{hE-hF-rel}
\end{gather}
where we introduce 
\begin{equation}\label{A-element}
  A=\,\sum_{s=1}^{p-1}\ffrac{(u_s(\q^{-s-1})-u_s(\q^{s-1}))K
        +\q^{s-1}u_s(\q^{s-1})-\q^{-s-1}u_s(\q^{-s-1})}{(\q^{s-1}
         -\q^{-s-1})u_s(\q^{-s-1})u_s(\q^{s-1})}\,
        u_s(K)\idem_s
\end{equation}
with $u_s(K)=\prod_{n=1,\;n\neq s}^{p-1}(K-\q^{s-1-2n})$, and
$\idem_s$ are the central primitive idempotents of $\UresSL2$ given in~\bref{app:idem}.

The comultiplication in $\LUresSL2$ is given
in~\eqref{Uq-comult-relations} for the $E$, $F$, and $K$ generators
and
\begin{gather}
  \Delta(e)=e\tensor1+K^p\tensor e
  +\ffrac{1}{[p-1]!} \sum_{r=1}^{p-1}\frac{\q^{r(p-r)}}{[r]}K^pE^{p-r}\tensor E^r
  K^{-r},\label{e-comult}\\
 \Delta(f)= f\tensor 1+K^p\tensor f+\ffrac{(-1)^p}{[p-1]!} 
  \sum_{s=1}^{p-1}\frac{\q^{-s(p-s)}}{[s]}K^{p+s}F^s\tensor F^{p-s}.\label{f-comult}
\end{gather}
An explicit form of $\Delta(h)=\half[\Delta(e),\Delta(f)]$ is very
bulky and we do not give it here. 

The antipode $S$ and the counity $\epsilon$ are given
in~\eqref{Uq-antipode}-\eqref{Uq-epsilon} and
\begin{gather}
  S(e)=-K^pe,\qquad S(f)=-K^pf,\qquad S(h)=-h,\\
  \epsilon(e)=\epsilon(f)=\epsilon(h)=0.\label{relations-end}
\end{gather}

\subsection{Irreducible $\LUresSL2$-modules}\label{sec:irreps}
An irreducible $\LUresSL2$-module $\repX^{\pm}_{s,r}$ is labeled by
$(\pm,s,r)$, with $1\leq s\leq p$ and $r\in\oN$, and has the highest
weights $\pm\q^{s-1}$ and $\frac{r-1}{2}$ with respect to $K$ and $h$
generators, respectively. The $sr$-dimensional module $\XX^{\pm}_{s,r}$
is spanned by elements $\stprp_{n,m}^{\pm}$, $0\leq n\leq s{-}1$,
$0\leq m\leq r{-}1$, where $\stprp_{0,0}^{\pm}$ is the highest-weight
vector and the left action of the algebra on $\XX^{\pm}_{s,r}$ is
given~by
\begin{align}
  K \stprp_{n,m}^{\pm} &=
  \pm \q^{s - 1 - 2n} \stprp_{n,m}^{\pm},\qquad
  &h\, \stprp_{n,m}^{\pm} &=  \half(r-1-2m)\stprp_{n,m}^{\pm},\label{basis-lusz-irrep-1}\\
  E \stprp_{n,m}^{\pm} &=
  \pm [n][s - n]\stprp_{n - 1,m}^{\pm},\qquad
  &e\, \stprp_{n,m}^{\pm} &=  m(r-m)\stprp_{n,m-1}^{\pm},\label{basis-lusz-irrep-2}\\
  F \stprp_{n,m}^{\pm} &= \stprp_{n + 1,m}^{\pm},\qquad
  &f\, \stprp_{n,m}^{\pm} &=  \stprp_{n,m+1}^{\pm},\label{basis-lusz-irrep-3}
\end{align}
where we set $\stprp_{-1,m}^{\pm}=\stprp_{n,-1}^{\pm}
=\stprp_{s,m}^{\pm} =\stprp_{n,r}^{\pm}=0$.


\subsection{Projective $\LUresSL2$-modules}\label{proj-mod}
We now recall subquotient structure of projective covers\footnote{A
projective cover of an irreducible module is a ``maximal''
indecomposable module that can be mapped onto the irreducible.}
$\PP^{\pm}_{s,r}$ over $\LUresSL2$ introduced in~\cite{[BFGT]}. The
$\PP^{\pm}_{s,r}$ module is the projective cover of
$\repX^{\pm}_{s,r}$, for $1\leq s\leq p-1$, and has the subquotient
structure~\eqref{schem-proj} on the left for $r=1$ and on the right
for $r\geq 2$, where one should  replace each irreducible subquotient or submodule
 $\repX_{s,r+2k-1}$ by
$\repX^{\pm}_{s,r+2k-1}$, and $\repX_{s,r+2k}$ by
$\repX^{\pm}_{s,r+2k}$, for any $k\geq1$.
The $\LUresSL2$ action is explicitly described in
App.~\bref{app:proj-mod-base}.

A ``half'' of these projective modules was identified in the
fusion algebra calculated in~\cite{[BFGT]} with  logarithmic or staggered
Virasoro modules realized in $\LM(1,p)$ models.

\begin{rem}\cite{[BFGT]} We note there are no additional
parameters distinguishing nonisomorphic indecomposable
$\LUresSL2$-modules with the same subquotient structure as
in~\eqref{schem-proj}.
\end{rem}

\subsection{Semisimple length of a module} Let $\modN$ be a
$\LUresSL2$-module.  We define a \textit{semisimple filtration} of
$\modN$ as a tower of submodules
\begin{equation*}
  \modN=\modN_0\supset\modN_1\supset\ldots\supset\modN_l=0
\end{equation*}
such that each quotient $\modN_i/\modN_{i+1}$ is semisimple.  The
number $l$ is called the \textit{length} of the filtration.  In the
set of semisimple filtrations of $\modN$, there exists a filtration
with the minimum length~$\ell$.  We call~$\ell$ the \textit{semisimple
length} of $\modN$. The semisimple length is also known as the Loewy
length and the semisimple quotients $\modN_i/\modN_{i+1}$ constitutes
the so-called Loewy layers: the first Loewy layer of a module $\modN$ is
$\modN/J(\modN)$, where $J(\modN)$ is the Jacobson radical of the
module $\modN$, the second Loewy layer involves taking a quotient of the
radical $J(\modN)$ by its own Jacobson radical and so on.

Evidently, an indecomposable module has the semisimple length not less
than~$2$.  Any semisimple module has the semisimple length~$1$.

\section{Classification of $\LUq$-modules\label{sec:LU-rep}}
To describe the category $\catC_p$ of finite-dimensional
$\LUresSL2$-modules, we recall first\--ex\-ten\-sion groups associated with
a pair of irreducible modules and give results of a computation of $n$-extensions
in~\bref{sec:exts}. This allows us to decompose the representation category
$\catC_p$ into full subcategories.
In~\bref{sec:M-W}, we construct a family of indecomposable modules
of the ``Feigin--Fuchs'' type. A classification theorem for the
category $\catC_p$ is presented in~\bref{thm:cat-decomp}.

\subsection{Extension groups}\label{sec:exts}
Here, we compute $n$-extensions between irreducible modules over
$\LUresSL2$ using Serre--Hochschild spectral sequences associated with
a filtration on $\LUresSL2$ given by the subalgebra $\UresSL2$ and the
quotient algebra $U s\ell(2)$. We then use the extensions in order to
construct four families of indecomposable modules in~\bref{sec:M-W}. 

Let $A$ and $C$ be left $\LUresSL2$-modules.  We say that a short
exact sequence of $\LUresSL2$-modules $0\to A\to B\to C\to 0$ is an
\textit{extension} of $C$ by $A$, and we let $\Ext(C,A)$ denote the
set of equivalence classes (see, e.g.,~\cite{[M]}) of first extensions
of $C$ by~$A$. Similarly, we denote $n$-extensions by $\Extn(C,A)$.

\begin{thm}\label{lem:extn-irrep}
  \addcontentsline{toc}{subsection}{\thesubsection. \ $n$-Extensions
    between irreducible modules} For $1 \leq s \leq p-1$ and
  $\alpha\,{=}\,\pm$, 
  the $n$-extension groups for $n>1$ are
    \begin{equation*}
    \Extn(\repX^{\alpha}_{s,1},\repX^{\alpha'}_{s',r'})\cong
    \begin{cases}
      \oC\,\delta_{\alpha',\alpha}\,\delta_{s',s}\,\delta_{r',n+1}, &n-\text{even},\\
      \oC\,\delta_{\alpha',-\alpha}\,\delta_{s',p-s}\,\delta_{r',n+1}, &n-\text{odd},
    \end{cases}
  \end{equation*}
and for $r>1$, we have
    \begin{equation*}
    \Extn(\repX^{\alpha}_{s,r},\repX^{\alpha'}_{s',r'})\cong
    \begin{cases}
      \oC\,\delta_{\alpha',\alpha}\,\delta_{s',s}\,, &\text{even}\,n<r, \;\text{and}\;  r'=r+2k, \;\text{with}\; -\frac{n}{2} \leq k\leq \frac{n}{2},\\
      \oC\,\delta_{\alpha',\alpha}\,\delta_{s',s}\,,
      &\text{even}\,n\geq r, \;\text{and}\;  r'=2k, \;\text{with}\; \frac{n-r+2}{2} \leq k\leq \frac{n+r}{2},\\
      \oC\,\delta_{\alpha',-\alpha}\,\delta_{s',p-s}\,, &\text{odd}\,n<r, \;\text{and}\; r'=r+2k+1,
         \; -\frac{n+1}{2} \leq k\leq \frac{n-1}{2},\\
      \oC\,\delta_{\alpha',-\alpha}\,\delta_{s',p-s}\,,
      &\text{odd}\,n\geq r, \;\text{and}\;  r'=2k+1, \; \frac{n-r+1}{2} \leq k\leq \frac{n+r-1}{2},\\
      0,&\;\text{otherwise},
    \end{cases}
  \end{equation*}
  where $\delta_{a,b}$ is the Kronecker symbol and when $k$ takes half-integer values we assume it goes with the step $1$. 
\end{thm}
\begin{proof}
We first recall~\cite{[FGST2]} that the space $\ExtnU$ of $n$-extensions
between irreducible modules over the subalgebra $\UresSL2$ is at most
$(n+1)$-dimensional and there exists a nontrivial $n$-extension only between
$\repX^{\pm}_s$ and $\repX^{\mp}_{p-s}$ for odd $n$ and between 
$\repX^{\pm}_s$ and $\repX^{\pm}_{s}$ for even $n$, where $1 \leq s\leq p-1$ and
we set $\repX^{\pm}_s{=}\repX^{\pm}_{s,1}|_{\UresSL2}$. Moreover,
there is an action of $\LUresSL2$ on projective resolutions for
simple $\UresSL2$-modules and this generates an action of the
quotient-algebra $U s\ell(2)$ on the corresponding cochain complexes and
their cohomologies. Therefore, for an irreducible module $\repX$ and an
$\UresSL2$-module $\modM$, all extension groups
$\ExtUbul(\repX,\modM)$ are $s\ell(2)$-modules. In particular, the
space $\ExtnU(\repX^{\pm}_s,\repX^{\mp}_{p-s})$ is the $(n+1)$-dimensional
irreducible $s\ell(2)$-module for odd $n$.

Next, to calculate the $n$-extension groups between the simple
$\LUresSL2$-modules, we use the Serre-Hochschild spectral sequence
with respect to the subalgebra $\UresSL2$ and the quotient algebra
$U s\ell(2)$. The spectral sequence is degenerate at the second term due
to the semisimplicity of the quotient algebra and we thus obtain
\begin{equation*}
\Extn(\repX^{\alpha}_{s,r},\repX^{\alpha'}_{s',r'}) = H^0(U s\ell(2),
\ExtnU(\repX^{\alpha}_{s,r},\repX^{\alpha'}_{s',r'})),
\end{equation*}
where the right-hand side is the vector space of the
$s\ell(2)$-invariants in the $s\ell(2)$-module
$\ExtnU(\repX^{\alpha}_{s,r},\repX^{\alpha'}_{s',r'})$. This module is
nonzero only in the cases $\alpha'=-\alpha$, $s'=p-s$ for odd~$n$ and $\alpha'=\alpha$, $s'=s$ for even $n$, and isomorphic to
the tensor product $\repX_{n+1}\tensor\repX_r\tensor\repX_{r'}$ of 
$s\ell(2)$-modules, where $\repX_r$ is the $r$-dimensional
irreducible module. Obviously, the tensor product decomposes as
\begin{equation}\label{extn-sl-inv}
\repX_{n+1}\tensor\repX_r\tensor\repX_{r'} \cong \bigoplus_{t=|r-r'|+1}^{r+r'-1}\bigoplus_{k=|t-n-1|+1}^{t+n} \repX_k
\end{equation}
and a simple counting of trivial
$s\ell(2)$-modules in the direct sum~\eqref{extn-sl-inv}
 completes the proof.
\end{proof} 

The reader can find an alternative proof of~\bref{lem:extn-irrep}
in~\bref{sec:extn-res-proof}. The proof is a direct calculation
involving projective resolutions. These resolutions also constitute one
of our results and they are described in App.~\bref{sec:resolutions}.

We note finally that taking $n=1$ in~\bref{lem:extn-irrep} gives an immediate consequence
obtained in~\cite{[BFGT]}.
\begin{cor}\label{lemma:exts}\cite{[BFGT]}\;
For $1 \leq s\leq p-1$, $r\in \oN$ and $\alpha,\alpha'\,{\in}\,\{+,-\}$,
  there are vector-space isomorphisms
  \begin{equation}\label{ext-eq}
    \Ext(\repX^{\alpha}_{s,r},\repX^{\alpha'}_{s',r'})\cong
    \begin{cases}
      \oC,\quad \alpha'=-\alpha, \; s'=p-s, \; r'=r\pm 1,\\
      0,\quad \text{otherwise}.
    \end{cases}
  \end{equation}
There are no nontrivial extensions between $\XX^{\pm}_{p,r}$ and any
irreducible module.
\end{cor}


\subsection{Indecomposable modules}\label{sec:M-W}
We now construct four infinite families of indecomposable modules over $\LUresSL2$.

Using~\bref{lemma:exts}, we can ``glue'' two irreducible modules
into an indecomposable module only in the case if the irreducibles
have opposite signs, different evenness in the $r$-index and the
sum of the two $s$-indexes is equal to $p$. Thereby, for $1\leq s \leq
p-1$ and integers $r,n\geq1$, we can introduce four types of
indecomposable modules of the semisimple length $2$ classified by their ``shapes'': $\modW$-, $\modM$-,
$\modN$-, and reversed-$\modN$ modules denoted by the symbol $\modNbar$. 
The modules are described in~\eqref{schem-W}-\eqref{schem-Nbar} by their subquotient structure, where $\repX_1{\longrightarrow}\repX_2$ denotes an extension
by an element from the space~$\Ext(\repX_1,\repX_2)$, with $\repX_1$ being an irreducible
subquotient and $\repX_2$ an irreducible submodule.   The subquotient structure uniquely defines these modules, up to an isomorphism,  due to the one-dimensionality of 
the first-extension groups~\eqref{ext-eq}.  We now turn to an explicit description of these modules in terms of bases and action.


$\smash{\boldsymbol{\modW^{{\pm}}_{s,r}(n)}}$: The module
  $\modW^{\pm}_{s,r}(n)$ has the subquotient structure~\eqref{schem-W}
  where  each irreducible subquotient
  $\repX_{s,r+2k}$ should be replaced by $\repX^{\pm}_{s,r+2k}$, and $\repX_{s,r+2k-1}$
  by $\repX^{\mp}_{s,r+2k-1}$, and $n$ is the number of the bottom
  modules (filled dots $\bullet$). We first describe the
  $\LUresSL2$-action on a basis in $\modW^{\pm}_{s,r}(1)$. The basis
  is spanned by $\{\botpr_{n,m}\}_{\substack{0\le n\le p-s-1\\0\le
  m\le r}} \cup\{\leftpr_{k,l}\}_{\substack{0\le k\le s-1\\0\le l\le
  r-1}} \cup\{\rightpr_{k,l}\}_{\substack{0\le k\le s-1\\0\le l\le
  r+1}}$ and identified with the corresponding submodule in
  $\PP^{\mp}_{p-s,r+1}$ explicitly described in
  App.~\bref{app:modW-1-base}. The modules $\modW^{\pm}_{s,r}(n)$,
  with $n>1$, are defined then by taking appropriate submodules in the
  direct sum of $n$ modules $\modW^{\pm}_{s,r+2k}(1)$:
\begin{equation*}
\modW^{\pm}_{s,r}(n) \subset \modW^{\pm}_{s,r}(1)\oplus
\modW^{\pm}_{s,r+2}(1)\oplus\dots
\oplus \modW^{\pm}_{s,r+2n-2}(1)
\end{equation*}
where we take a basis for the subquotient
$\repX^{\pm}_{s,r+2k}$ in $\modW^{\pm}_{s,r}(n)$ as the sum of the
bases in the two subquotients
$\repX^{\pm}_{s,r+2k}$ and $\repX^{\pm}_{s,r+2k}$ in the direct sum
$\modW^{\pm}_{s,r+2k-2}(1)\oplus\modW^{\pm}_{s,r+2k}(1)$. 
We give an explicit action
in App.~\bref{app:modW-base} in the example of the module
  $\modW^{\pm}_{s,r}(2)$.

$\boldsymbol{\modM^{\pm}_{s,r}(n)}$: The module
$\modM^{\pm}_{s,r}(n)$ is defined as the contragredient module to
 the $\modW^{\pm}_{s,r}(n)$ module, which means
 that  all the arrows in the diagram for $\modW^{\pm}_{s,r}(n)$ should
 be reversed in order to get a diagram for
 $\modM^{\pm}_{s,r}(n)$. This module
has the subquotient structure~\eqref{schem-M} 
where  each irreducible submodule $\repX_{s,r+2k}$ should be  replaced by $\repX^{\pm}_{s,r+2k}$,
 and $\repX_{s,r+2k-1}$ by $\repX^{\mp}_{s,r+2k-1}$. For later convenience, we give here its subquotient structure
 
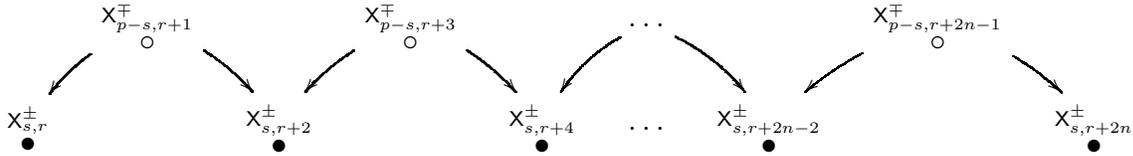
\begin{equation*}
   \xymatrix@=12pt{
     &\stackrel{\repX^{\mp}_{p-s,r+1}}{\circ}\ar@/_/[dl]
     \ar@/^/[dr]
     &&\stackrel{\repX^{\mp}_{p-s,r+3}}{\circ}\ar@/_/[dl]
     \ar@/^/[dr]
     &&\dots\ar@/_/[dl]\ar@/^/[dr]&
     &\stackrel{\repX^{\mp}_{p-s,r+2n-1}}{\circ}
     \ar@/_/[dl]\ar@/^/[dr]&\\
     \stackrel{\repX^{\pm}_{s,r}}{\bullet}&&
     \stackrel{\repX^{\pm}_{s,r+2}}{\bullet}
     &&\stackrel{\repX^{\pm}_{s,r+4}}{\bullet}
     &\dots&
     \stackrel{\repX^{\pm}_{s,r+2n-2}}{\bullet}&&
     \stackrel{\repX^{\pm}_{s,r+2n}}{\bullet}
   }
 \end{equation*}
 where $n$ is the number of the top modules (open dots $\circ$).
 The $\LUresSL2$-action on a basis is
 explicitly described in App.~\bref{app:modW-1-base} in the example of
 $\modM^{\pm}_{s,r}(1)$.  The modules $\modM^{\pm}_{s,r}(n)$,
  with $n>1$, are defined then by taking appropriate quotients of the
 direct sum of $n$ modules $\modM^{\pm}_{s,r+2k}(1)$ in accordance
 with the exact sequence
\begin{equation*}
0  \;\longrightarrow\; \bigoplus_{k=1}^{n-1}\XX^{\pm}_{s,r+2k} \;\longrightarrow\; \bigoplus_{k=0}^{n-1}\modM^{\pm}_{s,r+2k}(1)
  \;\longrightarrow\; \modM^{\pm}_{s,r}(n)  \;\longrightarrow\; 0
\end{equation*}
where  the image of each $\XX^{\pm}_{s,r+2k}$ under the embedding has
a basis which is the sum of the bases in the two  submodules
$\repX^{\pm}_{s,r+2k}$ and $\repX^{\pm}_{s,r+2k}$ in the direct sum
$\modM^{\pm}_{s,r+2k-2}(1)\oplus\modM^{\pm}_{s,r+2k}(1)$. 

$\boldsymbol{\modN^{\pm}_{s,r}(n)}$: The module
$\modN^{\pm}_{s,r}(n)$  is defined as the quotient of
  $\modM^{\pm}_{s,r}(n)$ by its submodule $\repX^{\pm}_{s,r+2n}$.
The $\modN^{\pm}_{s,r}(n)$  
  has the subquotient structure~\eqref{schem-N} 
where  each irreducible submodule $\repX_{s,r+2k}$ is replaced by $\repX^{\pm}_{s,r+2k}$,
 and $\repX_{s,r+2k-1}$ by $\repX^{\mp}_{s,r+2k-1}$
 and $n$ is the number of the top modules (open dots $\circ$) and at
  the same time the number of the bottom modules (filled dots
  $\bullet$). The $\LUresSL2$-action on a basis in
  $\modN^{\pm}_{s,r}(n)$ is explicitly described
  in~\bref{app:modW-1-base} in the example of the Weyl module
  $\modN^{\pm}_{s,r}(1)$.

$\smash{\boldsymbol{\modNbar^{{\pm}}_{s,r}(n)}}$: The module
  $\modNbar^{\pm}_{s,r}(n)$  is defined as the contragredient module to
 the $\modN^{\pm}_{s,r}(n)$ module defined just above, i.e. one should reverse all the arrows.
It has the subquotient structure~\eqref{schem-Nbar} where one should replace each irreducible subquotient $\repX_{s,r+2k}$ by $\repX^{\pm}_{s,r+2k}$, and $\repX_{s,r+2k-1}$ by $\repX^{\mp}_{s,r+2k-1}$,
  and $n$ is the number of the bottom modules (filled dots
  $\bullet$) and at the same time the number of the top modules (open
  dots $\circ$). The $\LUresSL2$-action on a basis in
  $\modNbar^{\pm}_{s,r}(n)$ is explicitly described
  in~\bref{app:modW-1-base} in the example of $\modNbar^{\pm}_{s,r}(1)$.

\medskip

The introduced four infinite series of indecomposable modules
$\modW^{{\pm}}_{s,r}(n)$, $\modM^{{\pm}}_{s,r}(n)$,
$\modN^{{\pm}}_{s,r}(n)$, and $\modNbar^{{\pm}}_{s,r}(n)$ are then
used in construction of the  projective
resolutions
and involved in a one-to-one correspondence with the Kac
modules over Virasoro.

\subsection{Classification theorem for the category $\catC_p$}\label{sec:cat-decomp}
Here, we describe the category $\catC_p$ of finite-dimensional
modules over $\LUresSL2$. We use results about
possible extensions between irreducible modules,
Thm.~\bref{lem:extn-irrep}, and the list of indecomposable
modules proposed above to state the following classification theorem.
\begin{Thm}\label{thm:cat-decomp}\mbox{}
Let $\catC_p$ denotes the category of finite-dimensional
  $\LUresSL2$-modules. Then, we have the following:
  \begin{enumerate}
  \item The category $\catC_p$  has the decomposition
    \begin{equation*}
      \catC_p=\bigoplus_{s=1}^{p-1}\Bigl(\catCpl(s)\oplus\catCmin(s)\Bigr)
      \oplus \bigoplus_{r\geq1}\Bigl(\catSpl(r)\oplus\catSmin(r)\Bigr),
    \end{equation*}
    where each direct summand is a full indecomposable subcategory.

  \item Each of the full subcategories $\catSpl(r)$ and~$\catSmin(r)$ is
    semisimple and contains precisely one irreducible module,
    $\repX^{+}_{p,r}$ and~$\repX^{-}_{p,r}$ respectively.
            
  \item Each full subcategory $\catCpl(s)$ contains the infinite
    family of irreducible modules~$\repX^{+}_{s,2r-1}$
    and~$\repX^{-}_{p-s,2r}$\,, with $r\in\oN$, and the following set
    of indecomposable modules:\enlargethispage{12pt}
    \begin{itemize}
    \item the projective modules $\PP^{+}_{s,2r-1}$ and
    $\PP^{-}_{p-s,2r}$\,, where $r\in\oN$;
      
    \item four series of indecomposable modules, for all integer
    $n\geq1$ and $r\in\oN$, given by 
      
      \begin{itemize}
      \item the $\modW^{+}_{s,2r-1}(n)$ and $\modW^{-}_{p-s,2r}(n)$
        modules and the contragredient to them $\modM^{+}_{s,2r-1}(n)$
        and $\modM^{-}_{p-s,2r}(n)$ modules;
        
      \item the $\modN^{+}_{s,2r-1}(n)$ and $\modN^{-}_{p-s,2r}(n)$
        modules and the contragredient to them $\modNbar^{+}_{s,2r-1}(n)$
        and $\modNbar^{\,-}_{p-s,2r}(n)$ modules.
      \end{itemize}
    \end{itemize}

  \item Each full subcategory $\catCmin(s)$ contains the infinite
    family of irreducible modules~$\repX^{+}_{s,2r}$
    and~$\repX^{-}_{p-s,2r-1}$\,, with $r\in\oN$, and the following set
    of indecomposable modules:\enlargethispage{12pt}
    \begin{itemize}
    \item the projective modules $\PP^{+}_{s,2r}$ and
    $\PP^{-}_{p-s,2r-1}$\,, where $r\in\oN$;
      
    \item four series of indecomposable modules, for all integer
    $n\geq1$ and $r\in\oN$, given by 
      
      \begin{itemize}
      \item the $\modW^{+}_{s,2r}(n)$ and $\modW^{-}_{p-s,2r-1}(n)$
        modules and the contragredient to them $\modM^{+}_{s,2r}(n)$
        and $\modM^{-}_{p-s,2r-1}(n)$ modules;
        
      \item the $\modN^{+}_{s,2r}(n)$ and $\modN^{-}_{p-s,2r-1}(n)$
        modules and the contragredient to them $\modNbar^{+}_{s,2r}(n)$
        and $\modNbar^{\,-}_{p-s,2r-1}(n)$ modules.
      \end{itemize}
    \end{itemize}
    This \textbf{exhausts} the list of indecomposable modules in $\catCpl(s)$ and $\catCmin(s)$.
  \end{enumerate}
\end{Thm}

 The strategy of the proof is
as follows. 
We first note that all projective modules
in~$\lc_p$ are injective modules.  This information suffices to ensure
that indecomposable modules with semisimple (Loewy) length~$3$ are projective
modules and that there are no modules with semisimple length~$4$ or
more.  Therefore, to complete the proof of~\bref{thm:cat-decomp}, it remains
to classify indecomposable modules with semisimple length~$2$.  We do
this in~\bref{sec:semisimple-2} using a correspondence between modules
with semisimple length $2$ and indecomposable representations of the
 quivers $A_N$, for appropriate $N$.

We now turn to a proof of 
Thm.~\bref{thm:cat-decomp}. We remind first the following fact
easily established using the identity $\Extn(\modPr,\modM)=0$ for a
projective module $\modPr$ and any module~$\modM$.
\begin{prop}\cite{[BFGT]}\mbox{}
\begin{enumerate}
\item
Every indecomposable $\LUresSL2$-module with the semisimple length~$3$
is isomorphic to $\PP^{\pm}_{s,r}$\,, for some $s\in \{1, 2, \dots,
p-1\}$ and some finite $r\in\oN$.
\item
There are no indecomposable modules with the semisimple length greater
than $3$.
\end{enumerate}
\end{prop}

\subsection{Modules with semisimple length $2$}\label{sec:semisimple-2}
To complete the proof of the parts (3) and (4) in~\bref{thm:cat-decomp}, it remains to
classify finite-dimensional modules with the semisimple length
$\ell\,{=}\,2$. We restrict our classification to the subcategory
$\catC_p^+$ because the full category $\catC_p$ is decomposed as
$\catC_p^+\oplus\catC_p^-$ where the two summands are equivalent as
abelian categories.
Indeed, the functor $\nu:\catCpl_p\to\catCmin_p$
mentioned in the introduction section and defined
by $\XX^{+}_{s,2r-1}\mapsto\XX^{-}_{s,2r-1}$, and
$\XX^{-}_{s,2r}\mapsto\XX^{+}_{s,2r}$, and similarly for all
indecomposable objects, gives an equivalence between the abelian
categories $\catCpl_p$ and $\catCmin_p$.

\subsubsection{The category $\lc^{(2)}(s)$} For $1 \leq s \leq p{-}1$,
let $\lc^{(2)}(s)$ be the full subcategory of $\lc^+(s)$ consisting of
$\LUresSL2$-modules with semisimple length $\ell\leq2$. The set of simple
objects in $\lc^{(2)}(s)$  consist of the infinite
    family of irreducible modules~$\repX^{+}_{s,2r-1}$
    and~$\repX^{-}_{p-s,2r}$\,, with $r\geq1$.  Obviously, any
module in $\lc^{(2)}(s)$ can be obtained
 either by
\begin{itemize}
\item the extension of a finite direct sum of semisimple modules
  $n_r\repX^{+}_{s,2r-1}$, with $r\geq 1$ and multiplicities $n_r\geq0$, by a direct sum of 
  $m_{r'}\repX^{-}_{p-s,2r'}$, with $r'\geq 1$ and multiplicities $m_{r'}\geq0$, via a direct sum of
  $x^+\!\in\!\Ext(n_r\repX^{+}_{s,2r-1},m_{r'}\repX^{-}_{p-s,2r'})$
\end{itemize}
or by
\begin{itemize}
\item the extension of a finite direct sum of semisimple modules
  $m_{r'}\repX^{-}_{p-s,2r'}$ by a direct sum of modules
  $n_r\repX^{+}_{s,2r-1}$ via a direct sum of
  $x^-\,{\in}\,\Ext(m_{r'}\repX^{-}_{p-s,2r'},n_r\repX^{+}_{s,2r-1})$.
\end{itemize}

For any finite set of  multiplicities
  $\{m_{r},n_{r};\, r\geq1, \,m_{r},n_{r}\geq0\}$, we choose an extension $\morX^+\,{\in}\,\Ext\left(\bigoplus_{r\geq1}n_{r}\repX^{+}_{s,2r-1},
  \bigoplus_{r\geq1} m_{r}\repX^{-}_{p-s,2r}\right)$. Let
$\modI_{\morX^{+}}(n_r,m_r)\in\ob(\catC^{(2)}(s))$ denotes the 
  module
defined by the extension $\morX^+$,
\begin{equation}\label{general-gluing}
\modI_{\morX^{+}}(n_r,m_r):\qquad  \bigoplus_{r\geq1} n_{r}\repX^{+}_{s,2r-1}
  \xrightarrow{\morX^{+}}\bigoplus_{r\geq1} m_{r}\repX^{-}_{p-s,2r},
\end{equation}
where we denote the dependence on the set
of the multiplicties $\{n_{r}, m_{r}\}$ in the round brackets, and omit braces and range for the $r$-index for brevity.
We define the modules
$\modI_{\morX^{-}}(m_r,n_r)$ similarly taking $\morX^-\,{\in}\,\Ext\left(\bigoplus_{r\geq1}m_{r}\repX^{-}_{p-s,2r},
  \bigoplus_{r\geq1} n_{r}\repX^{+}_{s,2r-1}\right)$. 
 The
  modules $\modW^{+}_{s,2r'-1}(n)$
  and $\modNbar^{{+}}_{s,2r'-1}(n)$ introduced
  in~\bref{sec:M-W}
are particular cases of the
  $\modI_{\morX^{+}}(n_r,m_r)$ modules: the $\modNbar^{+}_{s,2r'-1}(n)$
  corresponds to the multiplicities
  $n_{r}=m_{r}=1$ in~\eqref{general-gluing} for $r' \leq r\leq r'+n-1$ and they are zero otherwise.

Using~\bref{lemma:exts}, we note that a module  $\modI_{\morX^{\pm}}(n_r,m_{r'})$ is a
direct sum of two or more indecomposables if  there exist two (or more) subsets,
each indexed by $(r,r')$, in the set of non-zero values of
the multiplicities $n_r$, $m_{r'}$ such that they are separated  in the $r$- or $r'$-index by the value $2$ or more.
In order to calssify all indecomposable modules, we will thus restrict
to the following  choice of the multiplicities  $n_r$, $m_{r'}$: they are non-zero for all numbers in
regions
$1\leq k\leq r\leq k'$ and $2\leq l \leq
r'\leq l'$, where $l=k\pm1$ and $l'=k'\pm1$, and the multiplicities are zero otherwise.

We define next  full subcategories $\lc^{(2),+}_{n}(s)$ and
$\lc^{(2),-}_{n}(s)$, for $n\geq1$,
generating a filtration of the category $\lc^{(2)}(s)$ as follows. Isomorphism classes of simple objects from $\lc^{(2),+}_n(s)$ consist of the set
$\{\repX^+_{s,2r-1},\repX^-_{p-s,2r};\, 1\leq r\leq n\}$.  Any object of
$\lc^{(2),+}_n(s)$ is either a semisimple module or a module $\modN$
such that $\modN/\modN_1=\bigoplus_{r=1}^n n_r\repX^{+}_{s,2r-1}$ for appropriate
$n_r\geq0$, where $\modN_1$ is the maximal semisimple submodule
(the socle); in other words, an object of $\lc^{(2),+}_n(s)$ is a
 module $\modI_{\morX^{+}}(n_r,m_r)$ introduced in~\eqref{general-gluing} with $n_r=m_r=0$ for $r>n$.
Objects
of $\lc^{(2),-}_n(s)$ are defined similarly with
$\modN/\modN_1=\bigoplus_{r=1}^n m_r\repX^{-}_{p-s,2r}$, with some $m_r\geq0$.  We note that
we have the filtration by full abelian subcategories
\begin{equation*}
\lc^{(2),\pm}_1(s)\,\subset\,\lc^{(2),\pm}_2(s)\,
\subset\lc^{(2),\pm}_3(s)\subset\dots\subset\lc^{(2)}(s),
\end{equation*}
with $\ob(\lc^{(2)}(s))\,{=}\,
{\cup}_{n\geq1}\ob(\lc^{(2),\pm}_n(s))$.

We now reduce the classification of modules with semisimple length~$2$ in each $\lc^{(2),\pm}_{n}(s)$
to the classification of indecomposable representations of the
$A_{2n}$-type  quivers $\qK_{2n}$.  The reader is referred to~\cite{[CB],[ARS]}
and Appendix~\bref{app:quivers} for the necessary facts about quivers.

\begin{lemma}\label{lemma:mod-An}
  Each of the abelian categories $\lc^{(2),+}_n(s)$ and $\lc^{(2),-}_n(s)$ is
  equivalent to the category $\Rep(\qK_{2n})$ of representations of the
  $A_{2n}$-type quiver $\qK_{2n}$.
\end{lemma}
\begin{proof}
  The lemma is based on an observation that morphisms $\varepsilon_r$, with $1\leq r\leq 2n-1$,
   together with objects in
  \begin{multline*}
  \repX^+_{s,1}\xrightarrow{\varepsilon_1}\modM^+_{s,1}(1)\xleftarrow{\varepsilon_2}\repX^+_{s,3}
  \xrightarrow{\varepsilon_3}\modM^+_{s,3}(1)\xleftarrow{\varepsilon_4}\dots\xleftarrow{\varepsilon_{2n-4}}
  \repX^+_{s,2n-3}\xrightarrow{\varepsilon_{2n-3}}\\
  \xrightarrow{\varepsilon_{2n-3}}\modM^+_{s,2n-3}(1)\xleftarrow{\varepsilon_{2n-2}}
  \repX^+_{s,2n-1}\xrightarrow{\varepsilon_{2n-1}}\modN^+_{s,2n-1}(1)
  \end{multline*}
  make up a quiver $\qK_{2n}$ in the category $\lc^{(2),-}_n(s)$; 
  a similar collection of objects and morphisms (one should replace each $\repX^+_{s,2r-1}$ by $ \repX^-_{p-s,2r}$, $\modM^+_{s,2r-1}(1)$ by $\modM^-_{p-s,2r}(1)$, and $\modN^+_{s,2n-1}(1)$ by $\modNbar^+_{s,1}(1)$ and take all morphisms between these objects) make up a quiver $\qK_{2n}$ in the category $\lc^{(2),+}_n(s)$.
  We take then the functors of $\justHom$ to each of the two categories, $\lc^{(2),\pm}_n(s)$ and $\Rep(\qA)$, to
  establish an equivalence. 
  The
  equivalence, e.g., between the categories $\lc^{(2),-}_n(s)$ and
  $\Rep(\qA)$ is given by the functor $\FunF$ that acts on objects as
  \begin{equation}\label{FunF-act}
    \FunF(\modI_{\morX^-}(m_r,n_r))=((V_1,V_2,\dots,V_{2n}),f_{j,j\pm1}),
  \end{equation}
  where
  \begin{align*}
    &V_{2r-1}=\Hom(\repX^+_{s,2r-1},\modI_{\morX^{-}}(m_r,n_r))\,{=}\,\oC^{n_r},&\quad 1\leq &r\leq n,\\
       &V_{2r}=\Hom(\modM^+_{s,2r-1}(1), 
    \modI_{\morX^{-}}(m_r,n_r))\,{=}\,\oC^{m_r},&\quad 1\leq &r\leq n-1, \\
   &V_{2n}=\Hom(\modN^+_{s,2n-1}(1),
    \modI_{\morX^{-}}(m_r,n_r))\,{=}\,\oC^{m_n}&
  \end{align*}
  and, for linearly independent homomorphisms
  $\varepsilon_{2r-1},\varepsilon_{2r}
  \in\Hom(\repX^+_{s,2r\pm1},\modM^+_{s,2r-1}(1))\,{=}\,\oC$, the linear maps
  $f_{2r,2r\pm1}\in\HomC(V_{2r},V_{2r\pm1})$ are defined as
  \begin{equation*}
    f_{2r,2r-1}(\varphi)=\varphi\circ\varepsilon_{2r-1},\quad
        f_{2r,2r+1}(\varphi)=\varphi\circ\varepsilon_{2r},
  \end{equation*}
  for each $\varphi\in V_{2r}$, and with the natural action on morphisms.
  
  The existence of a functor $\FunG$ such that both $\FunG\FunF$ and
  $\FunF\FunG$ are the identity functors is evident from the
  definitions of the categories $\lc^{(2),-}_n(s)$ and $\Rep(\qK_{2n})$.  
\end{proof}


Propositions~\bref{lemma:mod-An} and~\bref{prop:repr-AN}
immediately imply the desired classification of finite-dimensional
$\LUresSL2$-modules with semisimple length $\ell=2$, thus completing
the proof of Thm.~\bref{thm:cat-decomp}.

\bigskip

We now turn to the most important part of the paper which presents
tensor product decompositions of all indecomposable modules over $\LUresSL2$.

\section{Tensor product decompositions}\label{sec:tens-prod}
To formulate the main result of the paper, we remind~\cite{[BFGT]}
that the tensor products between irreducible $\LUresSL2$-modules are
\begin{equation}\label{irred-fusion}
\repX^{\alpha}_{s_1,r_1}\tensor\repX^{\beta}_{s_2,r_2} =
\bigoplus_{\substack{r=|r_1-r_2|+1\\\step=2}}^{r_1+r_2-1}
\Bigr(\bigoplus_{\substack{s=|s_1-s_2|+1\\\step=2}}^{\substack{
\min(s_1 + s_2 - 1,\\ 2p - s_1 - s_2 - 1)}}\!\!\!\repX^{\alpha \beta}_{s,r}
\;\oplus\!\!\! \bigoplus_{\substack{s=2p - s_1 - s_2 +1\\\step=2}}^{p-\gamma_2}
\!\!\!\!\!\!\PP^{\alpha \beta}_{s,r}\Bigl),
\end{equation}
with $1\leq s_1, s_2 \leq p-1$, and $r_1, r_2\geq 1$, and $\alpha,\beta=\pm$.

We also refer the reader to~\eqref{not-1} which collect some notations we use
here intensively.

The following theorem states decompositions of tensor products of
irreducible modules $\repX^{\alpha}_{s_1,r_1}$ with the
(contragredient) Weyl modules $\modN^{\beta}_{s_2,r_2}(1)$ and $\modNbar^{\beta}_{s_2,r_2}(1)$.
\begin{Thm}\label{tens-prod-decomp-1}
For $1\leq s_1, s_2 \leq p-1$, $r_1, r_2\geq 1$, and $\alpha,\beta=\pm$, we have 
isomorphisms of $\LUresSL2$-modules
\begin{multline}\label{xn-tens-prod}
\repX^{\alpha}_{s_1,r_1}\tensor\modN^{\beta}_{s_2,r_2}(1) =
\bigoplus_{\substack{r=|r_1-r_2|+1\\\step=2}}^{r_1+r_2-1}
\,\bigoplus_{\substack{s=2p - s_1 - s_2 +
1\\\step=2}}^{p-\gamma_2}\!\!\!\!\!\!\PP^{\alpha
\beta}_{s,r} \,\oplus \bigoplus_{\substack{r=|r_1-r_2-1|+1\\\step=2}}^{r_1+r_2}
\, \bigoplus_{\substack{s= p + s_2 - s_1 +
1\\\step=2}}^{p-\gamma_1}\!\!\!\!\!\!\PP^{-\alpha \beta}_{s,r}\\
\oplus\bigoplus_{\substack{s=|s_1-s_2|+1\\\step=2}}^{\substack{\min(s_1 +
s_2 - 1,\\ 2p - s_1 - s_2 - 1)}}\!\!\!\begin{cases} &\modN^{\alpha
\beta}_{s,r_2-r_1+1}(r_1),\quad r_1 \le r_2\\ &\modW^{-\alpha
\beta}_{p-s,r_1-r_2}(r_2),\quad r_1 > r_2
\end{cases}
\end{multline}
and the tensor product with the  module contragredient to
$\modN^{\beta}_{s_2,r_2}(1)$ is
\begin{equation}\label{xnbar-tens-prod}
\repX^{\alpha}_{s_1,r_1}\tensor\modNbar^{\beta}_{s_2,r_2}(1) =
\bigr(\repX^{\alpha}_{s_1,r_1}\tensor\modN^{\beta}_{s_2,r_2}(1)\bigl)^*,
\end{equation}
where we set $\PP^*\equiv\PP$, $\modN^*\equiv\modNbar$, and $\modW^*\equiv\modM$.
\end{Thm}
\begin{proof}
 We consider first the tensor product in~\eqref{xn-tens-prod}. Let the
set $\{\stprp_{n',m'}\}$ denotes the basis in the first tensorand with
the action described in~\bref{sec:irreps}. The second tensorand has
the basis $\{\botpr_{n,m}\}\cup\{\rightpr_{k,l}\}$,
see~\eqref{N-one-basis}, with the $\LUresSL2$-action described in
App.~\bref{app:modW-1-base}. Taking the irreducible submodule
$\repX^{\beta}_{s_2,r_2}$ and subquotient
$\repX^{-\beta}_{p-s_2,r_2+1}$ of the module
$\modN^{\beta}_{s_2,r_2}(1)$, we consider the two (complementary)
subspaces $\repX^{\alpha}_{s_1,r_1}\tensor\repX^{\beta}_{s_2,r_2}$ and
$\repX^{\alpha}_{s_1,r_1}\tensor\repX^{-\beta}_{p-s_2,r_2+1}$ in the
tensor product space with the bases
$\{\stprp_{n',m'}\tensor\botpr_{n,m}\}$ and
$\{\stprp_{n',m'}\tensor\rightpr_{k,l}\}$, respectively, and decompose
them using~\eqref{irred-fusion}. Projectives obtained from these
tensor products are direct summands because any projective
$\LUresSL2$-module is also injective (the contragredient one to a
projective module) and is therefore a direct summand in any module
into which it is embedded. We thus obtain the decomposition
\begin{equation}\label{XN-decomp-gen}
\repX^{\alpha}_{s_1,r_1}\tensor\modN^{\beta}_{s_2,r_2}(1) = \modPP
\oplus \modII,
\end{equation}
where $\modPP$ is isomorphic to the direct sum over all projective
modules in the first row in~\eqref{xn-tens-prod} while the module
$\modII$ has the following relation in the Grothendieck ring
\begin{equation}\label{modII-Grot}
\bigl[\modII\bigr]  = \sum_{\substack{r=|r_1-r_2|+1\\\step=2}}^{r_1+r_2-1}
\sum_{\substack{s=|s_1-s_2|+1\\\step=2}}^{\substack{
\min(s_1 + s_2 - 1,\\ 2p - s_1 - s_2 -
1)}}\!\!\!\repX^{\alpha\beta}_{s,r} \; + 
\sum_{\substack{r=|r_1-r_2-1|+1\\\step=2}}^{r_1+r_2}
\sum_{\substack{s=|s_1+s_2-p|+1\\\step=2}}^{\substack{
\min(p + s_1 - s_2 - 1,\\ p - s_1 + s_2 - 1)}}\!\!\!\repX^{-\alpha\beta}_{s,r},
\end{equation}
where the first sum contributes to the socle $\soc(\modII)$ of
 $\modII$ because this direct sum is the
 submodule in the module $\repX^{\alpha}_{s_1,r_1}\tensor\repX^{\beta}_{s_2,r_2}$
 which is embedded into
 $\repX^{\alpha}_{s_1,r_1}\tensor\modN^{\beta}_{s_2,r_2}(1)$. We next
 show that $\modII$  turns
out to be a direct sum of indecomposables  and the first sum in~\eqref{modII-Grot}
 exhausts the socle of $\modII$, and moreover we show that the radical
 $\rad(\modII)\cong\soc(\modII)$, i.e. $\mathrm{top}(\modII)\cong\modII/\soc(\modII)$ is given
 by the second sum in~\eqref{modII-Grot}.

We give now explicit expressions for cyclic vectors generating the
module $\modII$ in~\eqref{XN-decomp-gen}.
We begin with expressions for highest weight vectors
$\toppr^{s,r}_{0,0}$ of the summands $\repX^{-\alpha \beta}_{s,r}$  in the second direct sum in~\eqref{modII-Grot},
\begin{equation}\label{hwv-top-exp}
\toppr^{p + s_1 - s_2 - 2n - 1,r_1 + r_2 - 2m}_{0,0} = \sum_{i = 0}^n \sum_{j = 0}^m A_i B_j\stprp_{i,j} \tensor \rightpr_{n-i,m-j},
\end{equation}
where $\max(0,s_1-s_2) \le n \le \min(s_1,p-s_2)-1$, and $0 \le m \le
\min(r_1,r_2+1)-1$, and the coefficients
\begin{equation}
A_i = (\alpha\q^{2n+s_2})^i\q^{-i^2}\frac{([n]!)^2[s_1-i-1]![p-s_2-n+i-1]!}{[s_1-n]![p-s_2-n]![i]![n-i]!},
\end{equation}
and
\begin{equation}\label{Bj-def}
B_j = (\alpha^p(-1)^{s_1})^j\frac{(m!)^2(r_1-j-1)!(r_2-m+j)!}{(r_1-m)!(r_2+1-m)!j!(m-j)!}.
\end{equation}

The highest weight vectors  $\botpr^{s,r}_{0,0}$ 
of the summands $\repX^{\alpha \beta}_{s,r}$ in the first direct sum
in~\eqref{modII-Grot} have a similar expression with the substitutions $\rightpr_{n-i,m-j} \to \botpr_{n-i,m-j}$,
 $s_2 \to p-s_2$ and $r_2 \to r_2-1$ to be applied in~\eqref{hwv-top-exp}.

Finally, cyclic vectors generating $\modII$ can be taken as
$\toppr^{s,r}_{0,1}\equiv f\toppr^{s,r}_{0,0}$ with the expression
\begin{equation}
\toppr^{p + s_1 - s_2 - 2n - 1,r_1 + r_2 - 2m}_{0,1} = \sum_{i = 0}^n \sum_{j = 0}^{m+1} A_i C_j\stprp_{i,j} \tensor \rightpr_{n-i,m+1-j},
\end{equation}
where $\max(0,s_1-s_2) \le n \le \min(s_1,p-s_2)-1$, and $0 \le m \le
\min(r_1,r_2+1)-1$.
\begin{equation}
C_j = B_j\alpha^p(-1)^{s_1-1} + B_{j-1}
\end{equation}
and $B_j$ is defined in~\eqref{Bj-def} and we set
$B_{-1}\equiv0$.
The following relation takes then place in the module $\modII$, or in the tensor product~\eqref{xn-tens-prod},
\begin{equation}
E\toppr^{s,r}_{0,1} \sim  \ffrac{r-1}{r}\botpr^{p-s,r+1}_{p-s-1,1} + \ffrac{1}{r}\botpr^{p-s,r-1}_{p-s-1,0},
\end{equation}
where $|s_1 - s_2|+1\le s \le \min(s_1+s_2,2p-s_1-s_2)-1$ and
$|r_1-r_2|+1 \le r \le r_1+r_2-1$, and $\botpr^{s,r}_{n,k}=F^n
f^k\botpr^{s,r}_{0,0}$, with $\botpr^{s,r}_{0,0}$ defined above
after~\eqref{Bj-def}.
Using the definition of $\modN_{s,r}(n)$ and $\modW_{s,r}(n)$ given in~\bref{sec:M-W}
(see also an explicit example of $\modW^{\pm}_{s,r}(2)$ in App. \bref{app:modW-base}),
 we state that the
module $\modII$ is isomorphic to the third direct sum (in the second
row) in~\eqref{xn-tens-prod}. Combining with~\eqref{XN-decomp-gen},
this finishes our proof of the decomposition~\eqref{xn-tens-prod}.

The decomposition~\eqref{xnbar-tens-prod} is obtained in a  very similar
way and involves the basis~\eqref{Nbar-one-basis} in
$\modNbar^{\beta}_{s_2,r_2}(1)$ with the action given also in
App.~\bref{app:modW-1-base}.
\end{proof}

We now turn to a more complicated case of tensor products of two
$\modN$- or $\modNbar$-type modules.
\begin{Thm}\label{tens-prod-decomp-2}
For $1\leq s_1, s_2 \leq p-1$, $r_1, r_2\geq 1$, and $\alpha,\beta=\pm$, we have 
isomorphisms of $\LUresSL2$-modules
\begin{multline}\label{NN-decomp}
\modN^{\alpha}_{s_1,r_1}(1)\tensor\modN^{\beta}_{s_2,r_2}(1) =
\bigoplus_{\substack{r=|r_1-r_2|+1\\\step=2}}^{r_1+r_2-1}\bigoplus_{\substack{s=|s_1-s_2|+1\\\step=2}}^{p-\gamma_2}\!\!\!\PP^{\alpha
\beta}_{s,r} \,\oplus \bigoplus_{\substack{s=s_1 + s_2 +
1\\\step=2}}^{p-\gamma_2}\!\!\!\!\!\!\PP^{\alpha
\beta}_{s,r_1+r_2+1}\,\oplus\\ \oplus
\bigoplus_{\substack{r=|r_1-r_2+\sg(s_2-s_1)|+1\\\step=2}}^{r_1+r_2}\bigoplus_{\substack{s=p-|s_1-s_2|+1\\\step=2}}^{p-\gamma_1}\!\!\!\PP^{-\alpha
\beta}_{s,r} \,\oplus \bigoplus_{\substack{s=|p-s_1-s_2|+1\\\step=2}}^{p -
|s_1 - s_2| - 1}\!\!\!\modN^{-\alpha \beta}_{s,r_1+r_2}(1)
\end{multline}
and the tensor product of  two  modules contragredient to
$\modN^{\pm}_{s,r}(1)$ is
\begin{equation}\label{NbarNbar-decomp}
\modNbar^{\alpha}_{s_1,r_1}(1)\tensor\modNbar^{\beta}_{s_2,r_2}(1) =
\bigl(\modN^{\alpha}_{s_1,r_1}(1)\tensor\modN^{\beta}_{s_2,r_2}(1)\bigr)^*,
\quad \text{with}\; \PP^*\equiv\PP, \; \modN^*\equiv\modNbar,
\end{equation}
and the
tensor product of $\modN$- and $\modNbar$-type modules decomposes as
\begin{multline}\label{NNbar-decomp}
\modN^{\alpha}_{s_1,r_1}(1)\tensor\modNbar^{\beta}_{s_2,r_2}(1)
=\bigoplus_{\substack{r=|r_1-r_2|+2\\\step=2}}^{r_1+r_2}\bigoplus_{\substack{s=|p-s_1-s_2|+1\\\step=2}}^{p-\gamma_1}\!\!\!\PP^{-\alpha
\beta}_{s,r}\;\oplus\!\!\! \bigoplus_{\substack{s=p-|s_1 - s_2|
+1\\\step=2}}^{p-\gamma_1}\!\!\!\!\!\!\PP^{-\alpha
\beta}_{s,|r_1-r_2|}\delta_{\sg(r_1-r_2),\sg(s_1-s_2)}\,\oplus\\
\oplus\bigoplus_{\substack{r=|r_1-r_2|+1\\\step=2}}^{r_1+r_2+\sg(p-s_1-s_2)}\bigoplus_{\substack{s=\min(s_1
+ s_2 + 1,\\ 2p - s_1 - s_2 +
1)\\\step=2}}^{p-\gamma_2}\!\!\!\PP^{\alpha
\beta}_{s,r}\,\oplus\bigoplus_{\substack{s=|p-s_1-s_2|+1\\\step=2}}^{p-|s_1 -
s_2| - 1}\!\!\!
\begin{cases}
&\modN^{-\alpha \beta}_{s,r_1-r_2}(1),\quad r_1>r_2\\
&\modNbar^{-\alpha \beta}_{s,r_2-r_1}(1),\quad r_2>r_1\\
&\repX^{\alpha\beta}_{p-s,1},\quad r_1=r_2.
\end{cases}
\end{multline}
\end{Thm}
\begin{proof}
We consider first the tensor product~\eqref{NN-decomp}. The first and
the second tensorands have the bases
$\{\botpr^1_{n,m}\}\cup\{\rightpr^1_{k,l}\}$ and
$\{\botpr^2_{n,m}\}\cup\{\rightpr^2_{k,l}\}$ respectively,
see~\eqref{N-one-basis}, with the $\LUresSL2$-action described in
App.~\bref{app:modW-1-base}. Taking the irreducible submodules
$\repX^{\alpha}_{s_1,r_1}$ and $\repX^{\beta}_{s_2,r_2}$ and
subquotients $\repX^{-\alpha}_{p-s_1,r_1+1}$ and
$\repX^{-\beta}_{p-s_2,r_2+1}$ of the modules, we consider the four
subspaces $\repX^{\alpha}_{s_1,r_1}\tensor\repX^{\beta}_{s_2,r_2}$,
$\repX^{\alpha}_{s_1,r_1}\tensor\repX^{-\beta}_{p-s_2,r_2+1}$,
$\repX^{-\alpha}_{p-s_1,r_1+1}\tensor\repX^{\beta}_{s_2,r_2}$, and
$\repX^{-\alpha}_{p-s_1,r_1+1}\tensor\repX^{-\beta}_{p-s_2,r_2+1}$ in
the tensor product space with the bases
$\{\botpr^1_{n,m}\tensor\botpr^2_{k,l}\}$,
$\{\botpr^1_{n,m}\tensor\rightpr^2_{k,l}\}$,
$\{\rightpr^1_{k,l}\tensor\botpr^2_{n,m}\}$, and
$\{\rightpr^1_{n,m}\tensor\rightpr^2_{k,l}\}$, respectively, and
decompose them using~\eqref{irred-fusion}.  Projective modules
obtained from these tensor products are direct summands because they
are also injective.  Irreducible
modules obtained from these tensor products contribute to submodules
or subquotients in indecomposable direct summands.  We thus obtain the
decomposition
\begin{equation}\label{NN-decomp-gen}
\modN^{\alpha}_{s_1,r_1}(1)\tensor\modN^{\beta}_{s_2,r_2}(1) =
\modPP \oplus \modII,
\end{equation}
where $\modPP$ is isomorphic to
the direct sums
\begin{equation*}
\bigoplus_{\substack{r=|r_1-r_2|+1\\\step=2}}^{r_1+r_2+\sg(p-s_2-s_1)}\bigoplus_{\substack{s=\min(s_1
+ s_2 + 1,\\ 2p- s_1 - s_2 +
1)\\\step=2}}^{p-\gamma_2}\!\!\!\!\!\!\PP^{\alpha \beta}_{s,r} \,\oplus
\bigoplus_{\substack{r=|r_1-r_2+\sg(s_2-s_1)|+1\\\step=2}}^{r_1+r_2}\bigoplus_{\substack{s=p-|s_1-s_2|+1\\\step=2}}^{p-\gamma_1}\!\!\!\PP^{-\alpha
\beta}_{s,r}
\end{equation*}
 over projective
modules (not all) in~\eqref{NN-decomp}
while the module
$\modII$ has the following relation in the Grothendieck ring
\begin{multline}\label{modII-Grot-2}
\bigl[\modII\bigr] = 
\sum_{\substack{r=|r_1-r_2|+1\\\step=2}}^{r_1+r_2-1}
\sum_{\substack{s=|s_1-s_2|+1\\\step=2}}^{\substack{\min(s_1 + s_2
    - 1,\\ 2p - s_1 - s_2 - 1)}} \!\!\!\repX^{\alpha \beta}_{s,r} +
\sum_{\substack{r=|r_1-r_2|+1\\\step=2}}^{r_1+r_2+1}
\sum_{\substack{s=|s_1-s_2|+1\\\step=2}}^{\substack{\min(s_1 + s_2
    - 1,\\ 2p - s_1 - s_2 - 1)}} \!\!\!\repX^{\alpha \beta}_{s,r}\\
+\sum_{\substack{r=|r_1-r_2-1|+1\\\step=2}}^{r_1+r_2}
\sum_{\substack{s=|s_1+s_2-p|+1\\\step=2}}^{\substack{\min(p + s_1 -
    s_2 - 1,\\ p - s_1 + s_2 - 1)}} \!\!\!\repX^{-\alpha \beta}_{s,r} +
\sum_{\substack{r=|r_1-r_2+1|+1\\\step=2}}^{r_1+r_2}
\sum_{\substack{s=|s_1+s_2-p|+1\\\step=2}}^{\substack{\min(p - s_1 +
    s_2 - 1,\\ p + s_1 - s_2 - 1)}} \!\!\!\repX^{-\alpha \beta}_{s,r}
\end{multline}
where the first sum obviously contributes to the socle $\soc(\modII)$
of $\modII$ because this direct sum is the
 submodule in the module $\repX^{\alpha}_{s_1,r_1}\tensor\repX^{\beta}_{s_2,r_2}$
 which is embedded into
 $\modN^{\beta}_{s_2,r_2}(1)\tensor\modN^{\beta}_{s_2,r_2}(1)$. We next
 show that $\modII$  turns
out to be a direct sum of indecomposables  and the first sum in~\eqref{modII-Grot-2}
 exhausts the socle (the first level) of $\modII$, and moreover we
 show below that the second level consists of the two last sums (in the second
 row) while the
 third level of $\modII$ is given
 by the second sum in the first row
 of~\eqref{modII-Grot-2}.

We note that if the summands contribute to a subquotient in an
indecomposable module of semisimple length not greater than $2$ (for
example to $\modN_{s,r}(n)$) then the Casimir element~\eqref{eq:casimir} has a diagonal
form on it. If the summands correspond to top subquotients in projective
modules then the Casimir element is non-diagonalizable on them. Thus, the
structure of the tensor product can be studied by
diagonalizability of the Casimir element
\begin{equation*}
  \cas=(\q-\q^{-1})^2 EF+\q^{-1}K+\q K^{-1}.
\end{equation*}

Our next strategy consists of two steps: (1) to study a Jordan cell decomposition of the matrix representig the Casimir element obtaining 
thus a projective module summand $\modPP_1$ in $\modII$ and then (2) to give cyclic vectors generating direct summands in $\modII$ which have the semisimple length not greater than $2$.
  
We assume in what follows that 
\begin{equation}\label{set-params}
0 \le n \le \min(p-s_1,p-s_2)-1, \quad
0 \le m_1 \le r_1+1, \quad 0 \le m_2 \le r_2+1.
\end{equation}

For any triplet $(n,m_1,m_2)$ with values from~\eqref{set-params}, let $T^{n}$ denotes a $(n+1)\times(n+1)$-matrix
representing the action of  $\cas$ restricted to the
subspace spanned by $\{\rightpr^1_{i,m_1}\tensor
\rightpr^2_{n-i,m_2}\}$, where $0 \le i \le n $ and we omit $m_1$
and $m_2$ indexes in the notation for this matrix, i.e. we consider
a decomposition (of the representation of) $\cas = \cas^{(d)}+\cas^{(n)}$ with a
diagonalizable part  $\cas^{(d)}=T^n$ while a non-diagonalizable part
$\cas^{(n)}$ will be given below.
$T^{n}$ is a three-diagonal matrix with the elements
\begin{equation}\label{cas-matr-form}
\begin{split}
&T^n_{i,i-1} = (\q-\q^{-1})^2\betabar[\ibar][\sbar_1-\ibar]\q^{\sbar_2-\sbar_1-2\nbar+4\ibar+2},\\
&T^n_{i,i} = \alphabar \betabar(\q^{-\sbar_1+2\ibar+1}(\q^{\sbar_2}+\q^{-\sbar_2})+\q^{\sbar_2-2\nbar+2\ibar-1}(\q^{\sbar_1}+\q^{-\sbar_1})-(\q+\q^{-1})\q^{\sbar_2-\sbar_1-2\nbar+4\ibar}),\\
&T^n_{i,i+1} =
  (\q-\q^{-1})^2\betabar[\nbar-\ibar][\sbar_2-\nbar-\ibar],\\
&T^n_{i,i\pm k} = 0,\qquad k\geq2.
\end{split}
\end{equation}
where $\sbar_1 = p-s_1$, $\sbar_2 = p-s_2$, $\nbar =
  n$, $\alphabar = -\alpha$, $\betabar = -\beta$, and $\ibar = i$.

The subspace $\{\rightpr^1_{i,m_1}\tensor
\rightpr^2_{n-i,m_2}\}$ is not invariant with respect to the Casimir
element action. We now describe the non-diagonalizable part
$\cas^{(n)}$. Depending on parameters,
there exist two cases:
\begin{enumerate}
\item
When $s_1+s_2-p+n < 0$, the smallest invariant subspace including $\{\rightpr^1_{i,m_1}\tensor \rightpr^2_{n-i,m_2}\}$ is
$V^{n} = \{\rightpr^1_{i,m_1}\tensor \rightpr^2_{n-i,m_2}\}\cup\{\rightpr^1_{j,m_1}\tensor
\botpr^2_{s_2+n-j,m_2-1}\}\cup\{\botpr^1_{k,m_1-1}\tensor \rightpr^2_{s_1+n-k,m_2}\}$,
where $n+1 \le j \le s_2+n $ and $0 \le k \le s_1-1 $.
Let $C_1^n$ denotes a $s_2\times s_2$ matrix representing $\cas$ within the subspace $\{\rightpr^1_{j,m_1}\tensor
\botpr^2_{s_2+n-j,m_2-1}\}$ and $C_2^n$ -- a $s_1\times s_1$ matrix representing $\cas$ within
$\{\botpr^1_{k,m_1-1}\tensor \rightpr^2_{s_1+n-k,m_2}\}$. They have the
elements as in~\eqref{cas-matr-form} with the substitutions $\sbar_1 = p-s_1$, $\sbar_2 = s_2$, $\nbar
= p-s_2+n$, $\ibar = j$, $\alphabar = -\alpha$, $\betabar = \beta$ and $\sbar_1 =
s_1$, $\sbar_2 = p-s_2$, $\nbar = p-s_1+n$, $\ibar = k$, $\alphabar = \alpha$, $\betabar =
-\beta$ respectively. A matrix representing the Casimir element within the whole
invariant subspace $V^n$ then takes the form
\begin{equation}\label{matrix_casimir_1}
  \xymatrix@=1pt{
    T^n&0&0\\
    \zeroRT{u_1}&C_1^n&0\\
    \zeroLB{u_2}&0&C_2^n
  }
\end{equation}
where $\zeroRT{u}$ denotes a matrix with all elements equal zero except
for the right top one which equals
\begin{equation}\label{u1-koeff}
u_1 = -(\q-\q^{-1})^2\beta m_2 \frac{(-1)^p[s_2]}{[p-1]!},
\end{equation}
and $\zeroLB{v}$
denotes a matrix with all elements equal zero except for the left bottom
one which equals 
\begin{equation}\label{u2-koeff}
u_2 = -(\q-\q^{-1})^2\beta m_1 \frac{(-1)^p[s_1]}{[p-1]!}\q^{s_1-s_2-2n+2}.
\end{equation}

\item
When $s_1+s_2-p+n \ge 0$, let $C_1^n$ similarly denotes a matrix
representing the Casimir element $\cas$ within the subspace
$\{\rightpr^1_{j,m_1}\tensor \botpr^2_{s_2+n-j,m_2-1}\}$ and $C_2^n$
-- a matrix representing $\cas$ within $\{\botpr^1_{k,m_1-1}\tensor
\rightpr^2_{s_1+n-k,m_2}\}$, where $n+1 \le j \le p-s_1-1 $ and
$s_2+s_1-p+n+1 \le k \le s_1-1 $. They have the elements as
in~\eqref{cas-matr-form} with $\sbar_1 = p-s_1$, $\sbar_2 = s_2$,
$\nbar = p-s_2+n$, $\ibar = j$, $\alphabar = -\alpha$, $\betabar =
\beta$ and $\sbar_1 = s_1$, $\sbar_2 = p-s_2$, $\nbar = p-s_1+n$,
$\ibar = k$, $\alphabar = \alpha$, $\betabar = -\beta$,
respectively. Let $B_1^n$ and $B_2^n$ also denotes matrices
representing $\cas$ within $\{\botpr^1_{m,m_1}\tensor
\botpr^2_{s_2+s_1-p+n-m,m_2-1}\}$ and $\{\botpr^1_{m,m_1-1}\tensor
\botpr^2_{s_2+s_1-p+n-m,m_2}\}$,
 respectively, where $0 \le m \le
s_1+s_2-p+n $. They have also the
elements~\eqref{cas-matr-form} with $\sbar_1 = s_1$, $\sbar_2 =
s_2$, $\nbar = s_1+s_2-p+n$, $\alphabar = \alpha$, $\betabar = \beta$
and $\ibar = m$. The matrix representing $\cas$ within the invariant
subspace $V^n = \{\rightpr^1_{i,m_1}\tensor
\rightpr^2_{n-i,m_2}\}\cup\{\rightpr^1_{j,m_1}\tensor
\botpr^2_{s_2+n-j,m_2-1}\}\cup\{\botpr^1_{k,m_1-1}\tensor
\rightpr^2_{s_1+n-k,m_2}\}\cup\{\botpr^1_{m,m_1}\tensor
\botpr^2_{s_2+s_1-p+n-m,m_2-1}\}\cup\{\botpr^1_{m,m_1-1}\tensor
\botpr^2_{s_2+s_1-p+n-m,m_2}\}$ takes then the block-structure
\begin{equation}\label{matrix_casimir_2}
  \xymatrix@C=2pt@R=16pt{
&&\\
&&\\
&A^n\;=&\\
&&
}
  \xymatrix@=1pt{
    T^n&0&0&0&0\\
    \zeroRT{u_1}&C_1^n&0&0&0\\
    0&\zeroRT{v_1}&B_1^n&0&0\\
    0&0&0&B_2^n&\zeroLB{v_2}\\
    \zeroLB{u_2}&0&0&0&C_2^n
  }
\end{equation}
where we set $v_1 = (\q-\q^{-1})^2\beta[p-s_1-n-1][s_1+s_2-p+n+1]$
and $v_2 =
(\q-\q^{-1})^2\beta[p-s_2-n-1][s_1+s_2-p+n+1]\q^{s_1+s_2-p+2n}$, and
$u_1$, $u_2$ are given in~\eqref{u1-koeff} and~\eqref{u2-koeff}.
\end{enumerate}

Thus, we have a set of subspaces $V^n$ which are invariant with respect to the Casimir element action
in the whole space of the tensor product.

Next, we assume a decomposition
\begin{equation}\label{I-decomp-gen}
\modII =
\modPP_1 \oplus \modII_1,
\end{equation}
where $\modPP_1$ is a maximum projective submodule, i.e. the direct sum
over all projective covers embedded in $\modII$,
while $\modII_1$ is a module of the semisimple length not greater than
$2$. We should note that
\begin{equation}[\mathrm{top}(\modPP_1)] \subset \sum_{\substack{r=|r_1-r_2|+1\\\step=2}}^{r_1+r_2+1}
\sum_{\substack{s=|s_1-s_2|+1\\\step=2}}^{\substack{\min(s_1 + s_2
    - 1,\\ 2p - s_1 - s_2 - 1)}} \!\!\!\repX^{\alpha \beta}_{s,r},
\end{equation}
where the sum was introduced in~\eqref{modII-Grot-2}, and every
highest weight vector from the sum (rigorously, in the corresponding direct sum of modules
considered as a quotient of $\modII$, of course)
appears  in the direct sum of spaces $\bigoplus_n V^n$.  Therefore, each projective
module that appears in $\modPP_1$ should contribute to a set of
Jordan cells (of rank $2$) in the matrix of $\cas$ on the space $\bigoplus_n V^n$. We thus need to
determine the Jordan cells structure on all subspaces $V^n$.

When $s_1+s_2 > p$, the Casimir element matrix has the block-structure
$A^n$~\eqref{matrix_casimir_2} for all subspaces $V^n$, with $0 \le n
\le \min(p-s_1,p-s_2)-1$. The set of eigenvalues of the $A^n$ matrix
coinsides with the union of the sets of eigenvalues of the blocks
$T^n$, $C_1^n$, $C_2^n$, $B_1^n$ and $B_2^n$. The eigenvalues of these
matrices are
\begin{align*}
&T^n:& \{\q^{2p-s_1-s_2-2k-1}+\q^{-(2p-s_1-s_2-2k-1)}, 0\leq k\leq n\},&\\
&C_1^n:& \{-\q^{s_2+p-s_1-2k-1}-\q^{-(s_2+p-s_1-2k-1)}, 0\leq k\leq p-s_1-n-2\},&\\
&C_2^n:& \{-\q^{s_1+p-s_2-2k-1}-\q^{-(s_1+p-s_2-2k-1)}, 0\leq k\leq p-s_2-n-2\},&\\
&B_1^n:& \{\q^{s_1+s_2-2k-1}+\q^{-(s_1+s_2-2k-1)}, 0\leq k\leq n+s_1+s_2-p\},&\\
&B_2^n:& \{\q^{s_1+s_2-2k-1}+\q^{-(s_1+s_2-2k-1)}, 0\leq k\leq n+s_1+s_2-p\}.&
\end{align*}

 We see that the eigenvalues
\begin{equation}\label{eigenvalues-set}
\{\q^{2p-s_1-s_2-2k-1}+\q^{-(2p-s_1-s_2-2k-1)}, 0\leq k\leq n\}
\end{equation}
are degenerated (see the first and two last rows above). 
Whether there are three eigenvectors corresponding to each triplet of
the degenerated eigenvalues or only two of them the Casimir element is
diagonalizable or not and the corresponding irreducible term is a
subquotient in a non-projective module or in a projective one. We
should note that there are other degenerated eigenvalues --
eigenvalues from the second and third rows which partially coincide,
namely the range $\{-\q^{p-|s_1-s_2|-2k-1}-\q^{-(p-|s_1-s_2|-2k-1)},
0\leq k\leq p-\max(s_1,s_2)-n-2\}$, but the corresponding subspaces
are direct summands as modules over the subalgebra generated by
$\cas$, therefore these eigenvalues can not correspond to Jordan
cells.  The degenerated eigenvalues from the last two rows (see the
fourth and fifth row above) from the range
$\{\q^{s_1+s_2-2k-1}+\q^{-(s_1+s_2-2k-1)}, 0\leq k\leq s_1+s_2-p-1\}$
comlementary to the range~\eqref{eigenvalues-set} correspond to
non-trivial Jordan cells but these cells contribute to projective
modules that are direct summands in the submodule
$\repX^{\alpha}_{s_1,r_1}\tensor\repX^{\beta}_{s_2,r_2}$ of the
tensor-product module and they are
therefore submodules in $\modPP$ and not in $\modII$
(see~\eqref{XN-decomp-gen}).

We analyze next the most degenerate range corresponding to~\eqref{eigenvalues-set}.
We claim that there are only two eigenvectors for each eigenvalue from the
set~\eqref{eigenvalues-set}: one vector is from the subspace
$\{\botpr^1_{m,m_1}\tensor\linebreak[0]\botpr^2_{s_2+s_1-p+n-m,m_2-1}\}$
and it coincides with the corresponding  eigenvector of the matrix $B_1^n$, and the
second one is from $\{\botpr^1_{m,m_1-1}\tensor
\botpr^2_{s_2+s_1-p+n-m,m_2}\}$ and coincides with the eigenvector of the matrix $B_2^n$,
where $0 \le m \le s_1+s_2-p+n $, and consequently there is a Jordan cell
contributing to projective modules in $\modII$. To prove that a third linearly independent eigenvector does not exist, we consider the equation
$(A^n-\lambda I)\mathsf{v} = 0$ (see~\eqref{matrix_casimir_2}), with
$\lambda$ being an eigenvalue from the set~\eqref{eigenvalues-set}. Since there exist at least two solutions of this equation, we
can decrease the number of variables by two.  The remaining equation
however has no a solution.
A similar analysis can be repeated for the case $p - n \le s_1+s_2 \le
p$.

 For $s_1+s_2 < p - n$, the only difference from the previous cases is that the matrix $T^n$ is not diagonalizable which lead to projective modules contributing to $\modPP$ while
the module $\modII$ has no projective modules as direct summands because (1) the socle of a projective direct summand in $\modII$ should be a proper subspace in the submodule $\repX^{\alpha}_{s_1,r_1}\tensor\repX^{\beta}_{s_2,r_2}$, which is spanned by vectors of the type $\botpr\tensor\botpr$, but (2) all Jordan cells of rank 2 in the Jordan form of the $A^n$ matrix given in~\eqref{matrix_casimir_1} are spanned by vectors which have an empty intersection with the submodule $\repX^{\alpha}_{s_1,r_1}\tensor\repX^{\beta}_{s_2,r_2}$ of the tensor-product module.

Combining with the $s\ell(2)$ content of $\modII$ in~\eqref{modII-Grot-2}, we finally get that the set of Jordan cells corresponds to
the decomposition~\eqref{I-decomp-gen} where
\begin{equation}
\modPP_1 = \bigoplus_{\substack{r=|r_1-r_2|+1\\\step=2}}^{r_1+r_2-1}
\bigoplus_{\substack{s=|s_1-s_2|+1\\\step=2}}^{\substack{\min(s_1 + s_2 - 1,\\ 2p - s_1 - s_2 - 1)}}\!\!\!\PP^{\alpha \beta}_{s,r},
\end{equation}
while the module $\modII_1$ has the following relation in the
Grothendieck ring
\begin{equation}
[\modII_1] = \sum_{\substack{s=|p-s_1-s_2|+1\\\step=2}}^{p - |s_1 - s_2| - 1}\!\!\!\repX^{-\alpha \beta}_{s,r_1+r_2}(1)+
\sum_{\substack{s=|s_1 - s_2| + 1\\\step=2}}^{\substack{\min(s_1 + s_2 - 1,\\ 2p - s_1 - s_2 - 1)}}\!\!\!\repX^{\alpha \beta}_{s,r_1+r_2+1}(1).
\end{equation}

Similarly to the proof in~\bref{tens-prod-decomp-1}, we obtain
\begin{equation}
\modII_1 = \bigoplus_{\substack{s=|p-s_1-s_2|+1\\\step=2}}^{p - |s_1 - s_2| - 1}\!\!\!\modN^{-\alpha \beta}_{s,r_1+r_2}(1),
\end{equation}
 where highest weight vectors and cyclic vectors of the subquotients have the expressions
\begin{align*}
&\toppr^{2p - s_1 - s_2 - 2n - 1}_{0,0} = \sum_{i = 0}^n A_i \rightpr^1_{i,0} \tensor \rightpr^2_{n-i,0},\\
&\toppr^{2p - s_1 - s_2 - 2n - 1}_{0,1} = \sum_{i = 0}^n A_i ((-\alpha)^p(-1)^{s_1-1}\rightpr^1_{i,0} \tensor \rightpr^2_{n-i,1}
+\rightpr^1_{i,1}\tensor \rightpr^2_{n-i,0}),
\end{align*}
where $\max(0,p-s_1-s_2) \le n \le \min(p-s_1,p-s_2)-1$, and the coefficients
\begin{equation}
A_i = (\alpha\q^{2n+s_2})^i\q^{-i^2}\frac{([n]!)^2[p-s_1-i-1]![p-s_2-n+i-1]!}{[p-s_1-n]![p-s_2-n]![i]![n-i]!}.
\end{equation}
Since $E\toppr^{2p - s_1 - s_2 - 2n - 1}_{0,0}=0$ and 
\begin{equation*}
E\toppr^{2p - s_1 - s_2 - 2n - 1}_{0,1} \in
\bigoplus_{\substack{s=|p-s_1-s_2|+1\\\step=2}}^{p - |s_1 - s_2| - 1}\!\!\!\repX^{-\alpha \beta}_{s,r_1+r_2}(1)
\end{equation*}
then $\toppr^{2p - s_1 - s_2 - 2n - 1}_{0,1}$ is a cyclic vector of $\modN(1)$ module (see App.~\bref{app:modW-1-base}).

Finally, a  decomposition of the whole tensor product is
\begin{align*}
&\modN^{\alpha}_{s_1,r_1}(1)\tensor\modN^{\beta}_{s_2,r_2}(1) =
\bigoplus_{\substack{r=|r_1-r_2|+1\\\step=2}}^{r_1+r_2-1}\bigoplus_{\substack{s=|s_1-s_2|+1\\\step=2}}^{\substack{\min(s_1
+ s_2 - 1,\\ 2p - s_1 - s_2 - 1)}}\!\!\!\PP^{\alpha \beta}_{s,r}\,\oplus
\bigoplus_{\substack{s=|p-s_1-s_2|+1\\\step=2}}^{p - |s_1 - s_2| -
1}\!\!\!\modN^{-\alpha \beta}_{s,r_1+r_2}(1)\,\oplus\\ &\oplus\!\!\!
\bigoplus_{\substack{r=|r_1-r_2|+1\\\step=2}}^{r_1+r_2+\sg(p-s_2-s_1)}\bigoplus_{\substack{s=\min(s_1
+ s_2 + 1,\\ 2p- s_1 - s_2 +
1)\\\step=2}}^{p-\gamma_2}\!\!\!\!\!\!\PP^{\alpha \beta}_{s,r} \,\oplus
\bigoplus_{\substack{r=|r_1-r_2+\sg(s_2-s_1)|+1\\\step=2}}^{r_1+r_2}\bigoplus_{\substack{s=p-|s_1-s_2|+1\\\step=2}}^{p-\gamma_1}\!\!\!\PP^{-\alpha
\beta}_{s,r}.
\end{align*}
which can be rewritten as in~\eqref{NN-decomp}.
The decompositions~\eqref{NbarNbar-decomp} and~\eqref{NNbar-decomp}
are obtained in a very similar way and we omit it.
\end{proof}

The above results allow us to decompose the tensor product of an arbitrary pair of indecomposable
modules over $\LUresSL2$.
\begin{Thm}\label{thm:tensor-third}
Tensor product of arbitrary two indecomposable  $\LUresSL2$-modules is obtained 
from the base tensor products in~\bref{tens-prod-decomp-1} and~\bref{tens-prod-decomp-2}, and the following list of rules:
\begin{enumerate}
 \item the tensor product of $\PP_{s,r}$ with an indecomposable 
module is isomorphic to the tensor product of $\PP_{s,r}$ with the direct
sum of all irreducible subquotients constituting the  indecomposable 
module.

\item An indecomposable module  with the semisimple length $2$ is the tensor product of irreducible
and simplest indecomposable modules for $s=1,\dots p-1$ and $r,n\in\oN$:
\begin{align}
\modN^{\alpha}_{s,r}(n) &= \repX^{+}_{1,n}\tensor\modN^{\alpha}_{s,r+n-1}(1),
&\modNbar^{\alpha}_{s,r}(n) &= \repX^{+}_{1,n}\tensor\modNbar^{\alpha}_{s,r+n-1}(1),\label{eq:modN-n}\\
\modW^{\alpha}_{s,r}(n) &= \repX^{-}_{1,r+n}\tensor\modN^{\alpha}_{p-s,n}(1),
&\modM^{\alpha}_{s,r}(n) &= \repX^{-}_{1,r+n}\tensor\modNbar^{\alpha}_{p-s,n}(1).\label{eq:modWM-n}
\end{align}
\end{enumerate}
\end{Thm}
\begin{proof}
 We note first
that the tensor product of a projective module with any indecomposable one must contain only projective modules.  The projectives obtained from this tensor product are direct
summands because any projective $\LUresSL2$-module is also injective
(the contragredient one to a projective module) and is therefore a
direct summand in any module into which it is embedded. This proves the first statement.

The second statement easily follows from the classification theorem~\bref{thm:cat-decomp} and
Thm.~\bref{tens-prod-decomp-1}. 
\end{proof}
This completes the description of the tensor structure on
$\catC_p$. Using the associativity and commutativity of the tensor
product decomposition for $\LUresSL2$, we give an exhaustive list of
tensor products in App.~\bref{app:prod_ind}.

\subsection{Generators} We give finally a set of generators in the
tensor category $\lc_p$:
\begin{equation*}
\XX^{\pm}_{1,1},\quad \XX^{+}_{1,2},\quad \XX^{+}_{2,1},\quad
\modN^{+}_{1,1}(1), \quad \modNbar^+_{1,1}(1).
\end{equation*}
That these objects generate the tensor category $\lc_p$ by successive
application of the tensor product follows easily from the previous three theorems given above.

\medskip

Since the construction of the tensor category $\catC_p$ is complete, we can
confirm a conclusion obtained in~\cite{[BFGT]} that the full
subcategory
\begin{equation*}
\catCpl_p=\bigoplus_{s=1}^{p-1}\catCpl(s)
      \oplus \bigoplus_{\text{odd}\, r\geq 1}\catSpl(r)
 \oplus \bigoplus_{\text{even}\, r\geq 2}\catSmin(r)
\end{equation*}
in the category $\catC_p=\catCpl_p\oplus\catCmin_p$ is closed under
the tensor product operation. All the base tensor products in $\catCpl_p$ are
collected in Thm.~\bref{thm:tens-prod-intro} where we use the notation
  $\XX_{s,r}\equiv\XX^{\alpha(r)}_{s,r}$, with the sign $\alpha(r)=(-1)^{r-1}$, and
  similar notations for all indecomposable modules.


\section{Conclusion}\label{sec:concl}
We have established, in~\bref{sec:results}, that each indecomposable $\LUresSL2$-module from
the category $\catC^+_p$ has an indecomposable counterpart in the
logarithmic $\LM(1,p)$ model, for any integer $p\geq 2$, which can be
described as a category $\mathcal{D}_p$ of representations of the
vertex operator algebra $\Vir$ corresponding to a quotient of the
universal enveloping of the Virasoro algebra with the central charge
$c_{1,p}=13-6/p-6p$.  This means that there exists a functor between
the two categories $\mathcal{F}:\catC^+_p\to\mathcal{D}_p$.  Moreover,
by direct comparison with~\cite{JR}, we see that $\mathcal{F}$ is a
tensor functor. This remarkable result allows us to conjecture that
$\catC^+_p$ and $\mathcal{D}_p$ are equivalent as tensor categories.
A possible way to prove the conjecture is to be reduced to a check
that the category $\mathcal{D}_p$ contains no more indecomposable
objects than $\catC^+_p$. This can be prooved by explicit comparison
of the $\EXT$ algebras for $\catC^+_p$ and $\mathcal{D}_p$.

We give finally several comments on relations between extension groups
for both tensor categories.
For each subcategory $\lc^+(s)$ with $1\leq s\leq p-1$, the basic fact
is that the space $\Extn$ of $n$-extensions between the irreducible
modules is at most one-dimensional (see~\bref{lem:extn-irrep}). We choose
 bases $\{x^+_{r,i}\}$ and $\{x^-_{r,i}\}$, with $i\in\{0,1\}$ and $r\geq1$,  in
the respective spaces
$\oC=\Ext(\repX^{+}_{s,2r-1},$\linebreak[0] $\repX^{-}_{p-s,2(r-i)})$ and
$\oC=\Ext(\repX^{-}_{p-s,2r},\repX^{+}_{s,2(r-i)+1})$, where we set
$x^+_{1,1}\equiv 0$. Next, the vector
space
\begin{equation*}
  \EXT_s=\bigoplus_{n\geq0}\,\bigoplus_{r,r'\geq1}
\Extn(\repX^{+}_{s,2r-1}
  \oplus\repX^{-}_{p-s,2r},
  \repX^{+}_{s,2r'-1}\oplus\repX^{-}_{p-s,2r'})
\end{equation*}
is an associative algebra with respect to the Yoneda product. 
We propose the algebraic structure of~$\EXT_s$.
\begin{Conj}\label{ext-alg}
  The algebra $\EXT_s$ is generated by $x^{\pm}_{r,i}$ with the
  defining relations
  \begin{gather*}
    x^+_{r,i} x^+_{r',j}=x^-_{r,i}x^-_{r',j}
    = x^-_{1,1}x^+_{1,0}=0 \\
    x^+_{r,0}x^-_{r,1}+x^+_{r+1,1}x^-_{r,0} = 0,\quad
    x^-_{r,0}x^+_{r+1,1}+x^-_{r+1,1}x^+_{r+1,0}=0,
  \end{gather*}
where $i,j\in\{0,1\}$  and $r,r'\geq1$.
\end{Conj}
Let $\EXT=\bigoplus_{s=1}^{p-1}\EXT_s$.
We note then that the derived category of representations of the algebra $\EXT$ is
equivalent to the derived category of $\catC^+_p$ 
and conjecturally to the derived category of $\mathcal{D}_p$.
An explicit calculation of the algebra of $\EXT$'s for the category $\mathcal{D}_p$
is a very important problem, which is waiting for its solution.

\medskip

\subsubsection*{Acknowledgments}
We are grateful to B.L.~Feigin, J.~Rasmussen, H.~Saleur,
A.M.~Semikhatov for valuable and stimulating discussions.  The work
of AMG was supported in part by the RFBR grant 10-01-00408, the
RFBR--CNRS grant 09-01-93105, and by the
``Landau'', ``Dynasty'' and ``Science Support'' foundations. The work
of PVB was supported in part by the RFBR grant 10-01-00408, the
RFBR--CNRS grant 09-01-93105 and ``Dynasty'' foundation.
IYuT is grateful to H.~Saleur for
kind hospitality in IPhT where a part of the work was made. IYuT was
supported in part by the
RFBR-CNRS grant 09-02-93106 and the RFBR Grant 08-02-01118.

\appendix

\section{Feigin-Fuchs modules}\label{feigin_fuchs}
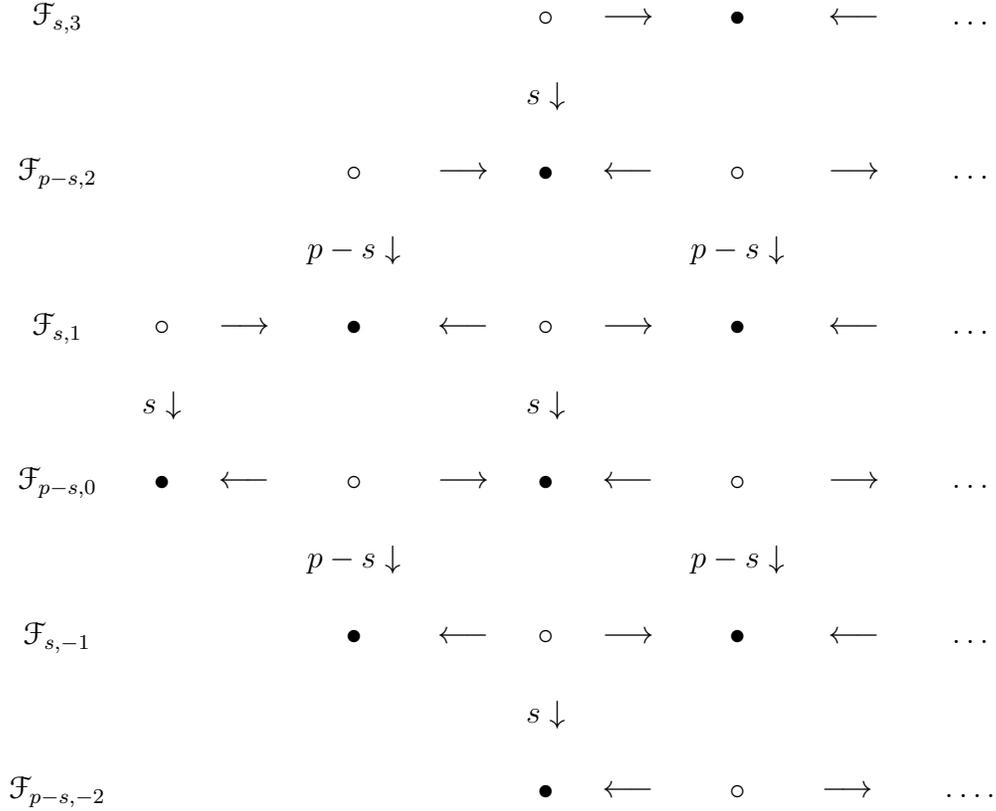
\begin{figure}
\begin{equation*}
  \xymatrix@C=6pt@R=10pt
  {
     \repF_{s,3}&&&&&\circ &\longrightarrow &\bullet &\longleftarrow \qquad \dots\\
     &&&&& s \downarrow &&&\\
     \repF_{p-s,2}&&&\circ &\longrightarrow &\bullet &\longleftarrow &\circ &\longrightarrow \qquad \dots\\
     &&& p-s \downarrow &&&& p-s \downarrow \\
     \repF_{s,1}&\circ &\longrightarrow &\bullet &\longleftarrow &\circ &\longrightarrow &\bullet &\longleftarrow \qquad \dots\\
     &s \downarrow &&&& s \downarrow &\\
     \repF_{p-s,0}&\bullet &\longleftarrow &\circ &\longrightarrow &\bullet &\longleftarrow &\circ  &\longrightarrow \qquad \dots\\
     &&& p-s \downarrow &&&& p-s \downarrow \\
     \repF_{s,-1}&&&\bullet &\longleftarrow &\circ &\longrightarrow &\bullet &\longleftarrow \qquad \dots\\
     &&&&& s \downarrow &&&\\
     \repF_{p-s,-2}&&&&&\bullet &\longleftarrow &\circ &\longrightarrow \qquad \dots.
  }
\end{equation*}
\caption{Felder complex of Feigin--Fuchs modules for a fixed $s$, 
where a down arrow like `$s\!\! \downarrow $' indicates the action of  $s^{\text{th}}$ power of the screening $F
=\oint\rme^{\alpha_-\varphi(z)}dz$ that defines a homomorphism from an upper
module with the second index being $n$ to the lower one with $n-1$. The $s$-morphisms and $(p-s)$-morphisms
are alternate. Filled dots $\bullet$ correspond to irreducible submodules -- they constitute the kernel 
of such homomorphisms. The  conformal
dimension of $\repF_{s,n}$ is $\Delta_{n,p-s} = \Delta_{1-n,s}$.}\label{felder-complex}
\end{figure} 
Here, we remind few simple facts about the well-known Feigin--Fuchs
modules which can be found in~\cite{FF}. 
 The Feigin--Fuchs module
$\repF_{s,n}$ over the Virasoro algebra $\Vir$ is the space generated
by all polynomials $P(\partial \phi)$ from the vertex-operator
$e^{(\frac{1-s}{2}\alpha_- + \frac{n}{2}\alpha_+)\phi}$, where $1 \le
s \le p$, $n\in \oZ$, and $\alpha_+=\sqrt{2p}$ and
$\alpha_-=-\sqrt{2/p}$. 
The (lowest) conformal
dimension in $\repF_{s,n}$ is $\Delta_{n,p-s} = \Delta_{1-n,s}$, where
$\Delta_{r,s}$ is defined in~\eqref{Delta}.
 When $s \ne p$, the $\repF_{s,n}$ has a
chain-type subquotient structure of one of the following patterns:
\begin{equation}
\circ \rightarrow  \bullet \leftarrow  \circ \rightarrow  \bullet \leftarrow  \dots
\end{equation}
or
\begin{equation}
\bullet \leftarrow  \circ \rightarrow  \bullet \leftarrow \circ \rightarrow \dots
\end{equation}
where filled and open dots $\bullet$ and $\circ$ correspond to
 irreducible submodules/subquotients.

For a fixed value of $s$, the Feigin--Fuchs modules form a Felder
complex~\cite{[F]} (see also~\cite{[FHST]}) shown in Fig.~\ref{felder-complex}.


\section{Projective $\LUresSL2$-modules}\label{app:proj-mod-base}
Here, we explicitly describe the $\LUresSL2$ action in the projective
module $\PP^{\pm}_{s,r}$. Let $s$ be an integer $1\leq s\leq p-1$ and
$r\in\oN$.  

For $r > 1$, the projective module $\PP^{\pm}_{s,r}$ has the basis
\begin{equation}\label{left-proj-basis-plus}
  \{\toppr_{n,m},\botpr_{n,m}\}_{\substack{0\le n\le s-1\\0\le m\le r-1}}
  \cup\{\leftpr_{k,l}\}_{\substack{0\le k\le p-s-1\\0\le l\le
  r-2}}
  \cup\{\rightpr_{k,l}\}_{\substack{0\le k\le p-s-1\\0\le l\le
  r}},
\end{equation}
where $\{\toppr_{n,m}\}_{\substack{0\le n\le s-1\\0\le m\le
    r-1}}$ is the basis
corresponding to the top module in~\eqref{schem-proj},\\
$\{\botpr_{n,m}\}_{\substack{0\le n\le s-1\\0\le m\le r-1}}$
to the bottom, $\{\leftpr_{k,l}\}_{\substack{0\le k\le
    p-s-1\\0\le l\le r-2}}$ to the left, and
$\{\rightpr_k\}_{\substack{0\le k\le p-s-1\\0\le l\le r}}$ to
the right module. 

For $r=1$, the basis does not contain
$\{\leftpr_{k,l}\}_{\substack{0\le k\le p-s-1\\0\le l\le
r-2}}$ terms and we imply $\leftpr_{k,l}~\equiv~0$ in the
action.  The $\LUresSL2$-action on $\PP^{\pm}_{s,r}$ is given by
\begin{align*}
  K\toppr_{n,m}&=\pm\q^{s-1-2n}\toppr_{n,m}, \quad 0\le n\le s-1,\quad 0\le m\le r-1,\\
  K\leftpr_{k,m}&=\mp\q^{p-s-1-2k}\leftpr_{k,m}, \quad 0\le k\le p-s-1,\quad 0\le m\le r-2,\\
  K\rightpr_{k,m}&=\mp\q^{p-s-1-2k}\rightpr_{k,m}, \quad 0\le k\le p-s-1,\quad 0\le m\le r,\\
  K\botpr_{n,m}&=\pm\q^{s-1-2n}\botpr_{n,m}, \quad 0\le n\le s-1,\quad 0\le m\le r-1,\\
  E\toppr_{n,m}&=
  \begin{cases}
    \pm[n][s-n]\toppr_{n-1,m}\pm g\botpr_{n-1,m}, &1\le n\le s-1,\\
    \pm g\frac{r-m}{r}\rightpr_{p-s-1,m}\pm g\frac{m}{r}\leftpr_{p-s-1,m-1}, & n=0,\\
  \end{cases}
  \quad 0\le m\le r-1,\\
  E\leftpr_{k,m}&=
  \begin{cases}
	\mp[k][p-s-k]\leftpr_{k-1,m}, &1\le k\le p-s-1, \\
    \pm g(m-r+1)\botpr_{s-1,m}, & k=0,\\
  \end{cases}
  \quad 0\le m\le r-2,\\
  E\rightpr_{k,m}&=
  \begin{cases}
    \mp[k][p-s-k]\rightpr_{k-1,m}, &1\le k\le p-s-1,\\
    \pm gm\botpr_{s-1,m-1}, & k=0,\\
  \end{cases}
  \quad 0\le m\le r,\\
  E\botpr_{n,m}&= \pm[n][s-n]\botpr_{n-1,m}, \quad 1\le n\le
  s-1, \quad 0\le m\le r-1 \quad (\botpr_{-1,m}\equiv 0),\\
  F\toppr_{n,m}&=
  \begin{cases}
    \toppr_{n+1,m}, &0\le n\le s-2,\\
    \frac{1}{r}\rightpr_{0,m+1}-\frac{1}{r}\leftpr_{0,m}, & n=s-1 
    \quad(\leftpr_{0,r-1}\equiv0),\\
  \end{cases}
   \quad 0\le m\le r-1,\\
  F\leftpr_{k,m}&=
  \begin{cases}
    \leftpr_{k+1,m}, &0\le k\le p-s-2,\\
    \botpr_{0,m+1}, & k=p-s-1,\\
  \end{cases}
  \quad 0\le m\le r-2,\\
  F\rightpr_{k,m}&=
  \begin{cases}
    \rightpr_{k+1,m}, & 0\le k\le p-s-2,\\
    \botpr_{0,m}, & k=p-s-1\quad (\botpr_{0,r}\equiv0),\\
  \end{cases}
  \quad 0\le m\le r,\\
  F\botpr_{n,m}&= \botpr_{n+1,m}, \quad 1\le n\le s-1,
  \quad 0\le m\le r-1 \quad (\botpr_{s,m}\equiv 0).
\end{align*}
where $g=\frac{(-1)^{p}[s]}{[p-1]!}$.
In thus introduced basis, the $s\ell(2)$-generators $e$, $f$ and $h$
act in $\PP^{\pm}_{s,r}$ as in the direct sum
$\XX^{\pm}_{s,r}\oplus\XX^{\mp}_{p-s,r-1}\oplus\XX^{\mp}_{p-s,r+1}
\oplus \XX^{\pm}_{s,r}$
(see~\eqref{basis-lusz-irrep-1}-\eqref{basis-lusz-irrep-3}), where for
$r=1$ we set $\XX^{\mp}_{p-s,0}\equiv 0$.

\section{Indecomposable $\LUresSL2$-modules: examples}\label{app:indecomp-mod-base}
Here, we explicitly describe the $\LUresSL2$ action in the
modules $\modM^{\pm}_{s,r}(1)$, $\modW^{\pm}_{s,r}(1)$,
$\modN^{\pm}_{s,r}(1)$, $\modNbar^{\pm}_{s,r}(1)$, and $\modW^{\pm}_{s,r}(2)$.

\subsection{Modules with $n=1$}\label{app:modW-1-base}
 Let $s$ be an integer $1\leq s\leq p-1$ and
$r\in\oN$.  
We note that there are
subquotients $\modM^{\mp}_{p-s,r-1}(1)$ and $\modNbar^{\pm}_{s,r}(1)$
and submodules $\modW^{\mp}_{p-s,r-1}(1)$ and $\modN^{\pm}_{s,r}(1)$
in the $\PP^{\pm}_{s,r}$ module. We thus use the formulas
in App.~\bref{app:proj-mod-base} to describe these modules:
\begin{itemize}
\item the module $\modM^{\mp}_{p-s,r-1}(1)$ has the basis
\begin{equation}\label{M-one-basis}
  \{\toppr_{n,m}\}_{\substack{0\le n\le s-1\\0\le m\le r-1}}
  \cup\{\leftpr_{k,l}\}_{\substack{0\le k\le p-s-1\\0\le l\le
  r-2}}
  \cup\{\rightpr_{k,l}\}_{\substack{0\le k\le p-s-1\\0\le l\le
  r}},
\end{equation}

\item the module $\modW^{\mp}_{p-s,r-1}(1)$ has the basis
\begin{equation}\label{W-one-basis}
  \{\botpr_{n,m}\}_{\substack{0\le n\le s-1\\0\le m\le r-1}}
  \cup\{\leftpr_{k,l}\}_{\substack{0\le k\le p-s-1\\0\le l\le
  r-2}}
  \cup\{\rightpr_{k,l}\}_{\substack{0\le k\le p-s-1\\0\le l\le
  r}},
\end{equation}

\item the module $\modN^{\pm}_{s,r}(1)$ has the basis
\begin{equation}\label{N-one-basis}
  \{\botpr_{n,m}\}_{\substack{0\le n\le s-1\\0\le m\le r-1}}
  \cup\{\rightpr_{k,l}\}_{\substack{0\le k\le p-s-1\\0\le l\le
  r}},
\end{equation}

\item the module $\modNbar^{\pm}_{s,r}(1)$ has the basis
\begin{equation}\label{Nbar-one-basis}
  \{\toppr_{n,m}\}_{\substack{0\le n\le s-1\\0\le m\le r-1}}
  \cup\{\rightpr_{k,l}\}_{\substack{0\le k\le p-s-1\\0\le l\le
  r}},
\end{equation}
\end{itemize}
and the algebra action on these modules coincides with the action on
the space $\PP^{\pm}_{s,r}$ with uninvolved basis vectors set
identically to zero. The action on $\PP^{\pm}_{s,r}$ is given above
after~\eqref{left-proj-basis-plus}.

\subsection{A module with $n=2$}\label{app:modW-base}
The indecomposable module $\modW^{\pm}_{s,r}(2)$ has the basis
\begin{equation}\label{modW-two-basis}
  \{\botpr^1_{k,l}\}_{\substack{0\le k\le p-s-1\\0\le l\le r}}
  \cup\{\botpr^2_{k,l}\}_{\substack{0\le k\le p-s-1\\0\le l\le
  r+2}}
  \cup\{\leftpr_{k,l}\}_{\substack{0\le k\le s-1\\0\le l\le
  r-1}}
  \cup\{\midpr_{k,l}\}_{\substack{0\le k\le s-1\\0\le l\le
  r+1}}
  \cup\{\rightpr_{k,l}\}_{\substack{0\le k\le s-1\\0\le l\le
  r+3}}.
\end{equation}

The $\LUresSL2$-action on $\modW^{\pm}_{s,r}(2)$ is given by
\begin{align*}
  K\midpr_{k,l}&=\pm\q^{s-1-2k}\toppr_{k,l}, \quad 0\le k\le s-1,\quad 0\le l\le r+1,\\
  K\leftpr_{k,l}&=\pm\q^{s-1-2k}\leftpr_{k,l}, \quad 0\le k\le s-1,\quad 0\le l\le r-1,\\
  K\rightpr_{k,l}&=\pm\q^{s-1-2k}\rightpr_{k,l}, \quad 0\le k\le s-1,\quad 0\le l\le r+3,\\
  K\botpr^1_{k,l}&=\mp\q^{p-s-1-2k}\botpr^1_{k,l}, \quad 0\le k\le p-s-1,\quad 0\le l\le r,\\
  K\botpr^2_{k,l}&=\mp\q^{p-s-1-2k}\botpr^2_{k,l}, \quad 0\le k\le p-s-1,\quad 0\le l\le r+2,\\
  E\midpr_{k,l}&=
  \begin{cases}
    \pm[k][s-k]\midpr_{k-1,l}, &1\le k\le s-1,\\
    \mp g(l-r-2)\botpr^1_{p-s-1,l-1} \mp gl\botpr^2_{p-s-1,l}, & k=0,\\
  \end{cases}
  \quad 0\le l\le r+1,\\
  E\leftpr_{k,l}&=
  \begin{cases}
	\pm[k][s-k]\leftpr_{k-1,l}, &1\le k\le s-1, \\
    \mp g(l-r)\botpr^1_{p-s-1,l}, & k=0,\\
  \end{cases}
  \quad 0\le l\le r-1,\\
  E\rightpr_{k,l}&=
  \begin{cases}
    \pm[k][s-k]\rightpr_{k-1,l}, &1\le k\le s-1,\\
    \mp gl\botpr^2_{p-s-1,l-1}, & k=0,\\
  \end{cases}
  \quad 0\le l\le r+3,\\
  E\botpr^1_{k,l}&= \mp[k][p-s-k]\botpr^1_{k-1,l}, \quad 0\le k\le
  p-s-1, \quad 0\le l\le r \quad (\botpr^1_{-1,l}\equiv 0),\\
  E\botpr^2_{k,l}&= \mp[k][p-s-k]\botpr^2_{k-1,l}, \quad 0\le k\le
  p-s-1, \quad 0\le l\le r+2 \quad (\botpr^2_{-1,l}\equiv 0),\\
  F\midpr_{k,l}&=
  \begin{cases}
    \midpr_{k+1,l}, &0\le k\le s-2,\\
    \botpr^1_{0,l}+\botpr^2_{0,l+1}, & k=s-1\quad (\botpr^1_{0,r+1}\equiv0),\\
  \end{cases}
   \quad 0\le l\le r+1,\\
  F\leftpr_{k,l}&=
  \begin{cases}
    \leftpr_{k+1,l}, &0\le k\le s-2,\\
    \botpr^1_{0,l+1}, & k=s-1,\\
  \end{cases}
  \quad 0\le l\le r-1,\\
  F\rightpr_{k,l}&=
  \begin{cases}
    \rightpr_{k+1,l}, & 0\le k\le s-2,\\
    \botpr^2_{0,l}, & k=s-1\quad (\botpr^2_{0,r+3}\equiv0),\\
  \end{cases}
  \quad 0\le l\le r+3,\\
  F\botpr^1_{k,l}&= \botpr^1_{k+1,l}, \quad 0\le k\le p-s-1,
  \quad 0\le l\le r \quad (\botpr^1_{p-s,l}\equiv 0),\\
  F\botpr^2_{k,l}&= \botpr^2_{k+1,l}, \quad 0\le k\le p-s-1,
  \quad 0\le l\le r+2 \quad (\botpr^2_{p-s,l}\equiv 0).
\end{align*}
where $g=\frac{(-1)^{p}[s]}{[p-1]!}$.
In thus introduced basis, the $s\ell(2)$-generators $e$, $f$ and $h$
act in $\modW^{\pm}_{s,r}(2)$ as in the direct sum
$\XX^{\pm}_{s,r}\oplus \XX^{\pm}_{s,r+2}\oplus \XX^{\pm}_{s,r+4}
\oplus\XX^{\mp}_{p-s,r+1}\oplus\XX^{\mp}_{p-s,r+3}$,
see~\eqref{basis-lusz-irrep-1}-\eqref{basis-lusz-irrep-3}.

\section{Projective resolution and higher extension groups}\label{sec:resolutions}
Here, we construct projective
resolutions for
irreducible modules. These resolutions involve the modules
$\modW^{\pm}_{s,r}$ and $\modN^{\pm}_{s,r}$
introduced in~\bref{sec:M-W}.
 Inspection based on the definition of the projective covers in~\bref{proj-mod} with their subquotient structure~\eqref{schem-proj}, and the definition of indecomposable modules from~\bref{sec:M-W} shows that the mappings
defined in the following lemma give rise to a projective resolution.
\begin{Prop}\label{lem:proj-res-irrep}
  For each $1\,{\leq}\, s \,{\leq}\,p{-}1$ and $\alpha\,{=}\,\pm$, \\
  \mbox{}$\bullet$\; the module $\repX^{\alpha}_{s,1}$
  has the projective resolution
  \begin{equation}\label{proj-res-irrep-triv}
    \ldots\xrightarrow{\partial_4}\modPr^{-\alpha}_{p-s,4}\xrightarrow{\partial_3}\modPr^{\alpha}_{s,3}
    \xrightarrow{\partial_2}
    \modPr^{-\alpha}_{p-s,2}\xrightarrow{\partial_1}\modPr^{\alpha}_{s,1}
    \stackrel{\partial_0}{\surjection}\repX^{\alpha}_{s,1}
  \end{equation}
  where for even $n$ the boundary morphism is given by the throughout mapping
  \begin{equation*}
    \partial_n:\modPr^{\alpha}_{s,n+1}
    \surjection\modN^{-\alpha}_{p-s,n}(1)\embedding
   \modPr^{-\alpha}_{p-s,n}\,,
  \end{equation*}
   and for odd $n$, the $n$th term and the boundary morphism
  $\partial_n$ are given by changing $\alpha$ to $-\alpha$ and $s$ to $p{-}s$;\\
  \mbox{}$\bullet$\; for even $r\geq2$ the module $\repX^{\alpha}_{s,r}$
  has the projective resolution
  \begin{multline}\label{proj-res-irrep}
  \ldots\xrightarrow{\partial_{r+2}}
  \modPr^{-\alpha}_{p-s,3}\oplus\modPr^{-\alpha}_{p-s,5}\oplus\dots\oplus
    \modPr^{-\alpha}_{p-s,2r+1}\xrightarrow{\partial_{r+1}}
    \modPr^{\alpha}_{s,2}\oplus\dots\oplus\modPr^{\alpha}_{s,2r}
    \xrightarrow{\boldsymbol{\partial_r}}\\
    \xrightarrow{\boldsymbol{\partial_r}}\modPr^{-\alpha}_{p-s,1}\oplus\modPr^{-\alpha}_{p-s,3}\oplus\dots\oplus
    \modPr^{-\alpha}_{p-s,2r-1}\xrightarrow{\partial_{r-1}}\modPr^{\alpha}_{s,2}
    \oplus\modPr^{\alpha}_{s,4}\oplus\dots\oplus\modPr^{\alpha}_{s,2r-2}
	\xrightarrow{\partial_{r-2}}\\
    \xrightarrow{\partial_{r-2}}\ldots\xrightarrow{\partial_3}\modPr^{\alpha}_{s,r-2}
    \oplus\modPr^{\alpha}_{s,r}\oplus\modPr^{\alpha}_{s,r+2}
    \xrightarrow{\partial_2}\modPr^{-\alpha}_{p-s,r-1}\oplus
    \modPr^{-\alpha}_{p-s,r+1}\xrightarrow{\partial_1}\modPr^{\alpha}_{s,r}
    \stackrel{\partial_0}{\surjection}\repX^{\alpha}_{s,r}
      \end{multline}
      which consists of two parts separated by the homomorphism $\boldsymbol{\partial_{r}}$:
      \begin{itemize}
  \item[$\boldsymbol{*}$] on the right from $\boldsymbol{\partial_{r}}$, i.e. $n<r$, the $n$th term with even $n$ is given by
  \begin{equation*}
    \xrightarrow{\partial_{n+1}}
    \underbrace{\modPr^{\alpha}_{s,r-n}\oplus\modPr^{\alpha}_{s,r-n+2}\oplus\dots\oplus
      \modPr^{\alpha}_{s,r+n}}_{n+1}\xrightarrow{\partial_{n}}
  \end{equation*}
  with the boundary morphism given by the throughout mapping in
  \begin{equation*}
    \partial_n:\underbrace{\modPr^{\alpha}_{s,r-n}
\oplus\dots\oplus
      \modPr^{\alpha}_{s,r+n}}_{n+1}
    \surjection\modW^{\alpha}_{s,r-n}(n)\embedding
    \underbrace{\modPr^{-\alpha}_{p-s,r-n+1}
      \oplus\modPr^{-\alpha}_{p-s,r-n+3}
      \oplus\dots\oplus\modPr^{-\alpha}_{p-s,r+n-1}}_{n},
  \end{equation*}
  and for odd $n$, the $n$th term and the boundary morphism
  $\partial_n$ are given by changing $\alpha$ to $-\alpha$ and $s$ to $p{-}s$; 

  \item[$\boldsymbol{*}$] on the left from $\boldsymbol{\partial_{r}}$ in~\eqref{proj-res-irrep}, i.e. $n\geq r$, the $(r+k)$th term with even $k\geq0$ is
    \begin{equation*}
    \xrightarrow{\partial_{r+k+1}}
    \underbrace{\modPr^{\alpha}_{s,k+2}\oplus\modPr^{\alpha}_{s,k+4}\oplus\dots\oplus
      \modPr^{\alpha}_{s,k+2r}}_{r}\xrightarrow{\partial_{r+k}}
  \end{equation*}
  with the boundary morphism given by the throughout mapping in
  \begin{equation*}
    \partial_{r+k}:\underbrace{\modPr^{\alpha}_{s,k+2}
\oplus\dots\oplus
      \modPr^{\alpha}_{s,k+2r}}_{r}
    \surjection\modN^{-\alpha}_{p-s,k+1}(r)\embedding
    \underbrace{\modPr^{-\alpha}_{p-s,k+1}
      \oplus\modPr^{-\alpha}_{p-s,k+3}
      \oplus\dots\oplus\modPr^{-\alpha}_{p-s,2r+k-1}}_{r},
  \end{equation*}
  and for odd $k$, the $(r+k)$th term and the boundary morphism
  $\partial_{r+k}$ are given by changing $\alpha$ to $-\alpha$ and $s$ to $p{-}s$; 
  \end{itemize}
    \mbox{}$\bullet$\; for odd $r\geq3$ the module $\repX^{\alpha}_{s,r}$
  has the projective resolution as in~\eqref{proj-res-irrep} with the substitution  $\alpha\to-\alpha$ and $s\to p{-}s$ in all terms and morphisms on the left from $\partial_{r-2}$.
\end{Prop}

\begin{rem}
One can easily obtain injectve resolutions dualising the statement in~\bref{lem:proj-res-irrep} -- by reversing all the homomorphisms and replacing all the modules of $\modW$- and $\modN$-types by corresponding conrtagredient ones, $\modM$- and $\modNbar$-modules.
\end{rem}

\subsection{Higher extensions via the resolution}\label{sec:extn-res-proof}
We now use the projective resolutions to calculate $\Extn$ between
simple modules. They are collected in~\bref{lem:extn-irrep}.
 
 The contravariant functor $\Hom(-,\repX^{\alpha'}_{s',r'})$ applied to the
  projective resolution~\eqref{proj-res-irrep-triv} of $\repX^{\alpha}_{s,1}$ described in~\bref{lem:proj-res-irrep} gives the cochain complex
  \begin{equation*}
  0\xrightarrow{\delta_0}\Hom(\modPr^{\alpha}_{s,1},\repX^{\alpha'}_{s',r'})
    \xrightarrow{\delta_1}\Hom(\modPr^{-\alpha}_{p-s,2},\repX^{\alpha'}_{s',r'})
    \xrightarrow{\delta_2}\Hom(\modPr^{\alpha}_{s,3},\repX^{\alpha'}_{s',r'})\xrightarrow{\delta_3}\ldots,
  \end{equation*}
 where an odd $n$th term is non-trivial only if $\alpha'=\alpha$, $s'=s$ and $r'=n$, and it is one-dimensional, while
  an even $n$th term is non-trivial only if $\alpha'=-\alpha$, $s'=p-s$ and $r'=n$, and it is also one-dimensional.
 We thus have
   that all coboundary morphisms $\delta_i=0$, $i\geq0$.   The cohomologies $\ker(\delta_{n+1})/\im(\delta_{n})$ of this complex 
  give then 
   $\Extn(\repX^{\alpha}_{s,1},\repX^{\alpha'}_{s',r'})$.

  Applying the contravariant functor $\Hom(-,\repX^{\alpha'}_{s',r'})$ to the
  projective resolution~\eqref{proj-res-irrep} of $\repX^{\alpha}_{s,r}$, for even $r$, gives the cochain
  complex\pagebreak[3]
  \begin{multline*}
    \bP:
    0\xrightarrow{\delta_0}\Homm(\modPr^{\alpha}_{s,r},\repX^{\alpha'}_{s',r'})
    \xrightarrow{\delta_1}
    \Homm(\modPr^{-\alpha}_{p-s,r-1}\oplus\modPr^{-\alpha}_{p-s,r+1},\repX^{\alpha'}_{s',r'})
    \xrightarrow{\delta_2}\ldots\xrightarrow{\delta_{r-1}}\\
    \xrightarrow{\delta_{r-1}}\Homm(\modPr^{-\alpha}_{p-s,1}\oplus\dots\oplus
    \modPr^{-\alpha}_{p-s,2r-1},\repX^{\alpha'}_{s',r'})
\xrightarrow{\boldsymbol{\delta_r}}\Homm(\modPr^{\alpha}_{s,2}\oplus\dots\oplus\modPr^{\alpha}_{s,2r},
\repX^{\alpha'}_{s',r'})\\
\xrightarrow{\delta_{r+1}}
\Homm(\modPr^{-\alpha}_{p-s,3}\oplus\modPr^{-\alpha}_{p-s,5}\oplus\dots\oplus
    \modPr^{-\alpha}_{p-s,2r+1},\repX^{\alpha'}_{s',r'})
\xrightarrow{\partial_{r+2}}\ldots
  \end{multline*}
  where we set $\Homm\equiv\Hom$ and  the $(n\!+\!1)$th term for even $n<r$ is given by
  \begin{equation}\label{extn-odd}
    \dots\;\xrightarrow{\delta_n}\Hom(
    \underbrace{\modPr^{\alpha}_{s,r-n}\oplus\modPr^{\alpha}_{s,r-n+2}\oplus\dots\oplus
      \modPr^{\alpha}_{s,r+n}}_{n+1},\repX^{\alpha'}_{s',r'})
    \xrightarrow{\delta_{n+1}}\;\dots
  \end{equation}
and for odd $n<r$, the $(n\!+\!1)$th term is 
  \begin{equation}\label{extn-even}
    \dots\;\xrightarrow{\delta_n}\Hom(
    \underbrace{\modPr^{-\alpha}_{p-s,r-n}
      \oplus\modPr^{-\alpha}_{p-s,r-n+2}
      \oplus\dots\oplus\modPr^{-\alpha}_{p-s,r+n}}_{n+1},\repX^{\alpha'}_{s',r'})
    \xrightarrow{\delta_{n+1}}\;\dots
  \end{equation}
  The term in~\eqref{extn-odd} is non-zero only if $s'=s$,
  $\alpha'=\alpha$ and $r'=r+2k$, with $-n/2 \leq k\leq n/2$, and is
  isomorphic to $\oC$ while the term in~\eqref{extn-even} is non-zero
  when $s'=p-s$, $\alpha'=-\alpha$ and $r'=r+2k+1$, with $-(n+1)/2
  \leq k\leq (n-1)/2$, and is also isomorphic to $\oC$. The two cases
  have zero intersection and we thus have that all coboundary
  morphisms $\delta_i=0$, $i\geq1$.  The cohomologies
  $\ker(\delta_{n+1})/\im(\delta_{n})$ of the complex $\bP$ give then the
  $n$-extension groups
  $\Extn(\repX^{\alpha}_{s,r},\repX^{\alpha'}_{s',r'})$ for all
  $n<r$. The higher-extension groups are calculated using similar
  analysis of the cohomologies of $\bP$ for $n\geq r$.

  To calculate $\Extn(\repX^{\alpha}_{s,r},\repX^{\alpha'}_{s',r'})$ for odd $r$, we proceed similarly using  the
  projective resolution for $\repX^{\alpha}_{s,r}$ in the odd-$r$ case described in~\bref{lem:proj-res-irrep}.

We note that the $n$-extensions between all irreducible modules computed using the projective resolutions are in remarkable coincidence
with the result~\eqref{extn-sl-inv} obtained by the direct calculation which uses the spectral sequence. 

\section{Tensor products of indecomposable modules}\label{app:prod_ind}
We use here Thm.~\bref{thm:tensor-third} and~\eqref{eq:modN-n}-\eqref{eq:modWM-n} to give  an exhaustive list of
tensor products of all indecomposable modules with the semisimple length not greater than $2$.
\begin{align*}
&\repX^{\alpha}_{s_1,r_1}\tensor\modN^{\beta}_{s_2,r_2}(n)=\left(\repX^{\alpha}_{s_1,r_1}\tensor\repX^{\beta}_{1,n}\right)\tensor\modN^+_{s_2,r_2+n-1}(1),\\
&\repX^{\alpha}_{s_1,r_1}\tensor\modNbar^{\beta}_{s_2,r_2}(n)=\left(\repX^{\alpha}_{s_1,r_1}\tensor\repX^{\beta}_{1,n}\right)\tensor\modNbar^+_{s_2,r_2+n-1}(1),\\
&\repX^{\alpha}_{s_1,r_1}\tensor\modW^{\beta}_{s_2,r_2}(n)=\left(\repX^{\alpha}_{s_1,r_1}\tensor\repX^{\beta}_{1,r_2+n}\right)\tensor\modN^-_{p-s_2,n}(1),\\
&\repX^{\alpha}_{s_1,r_1}\tensor\modM^{\beta}_{s_2,r_2}(n)=\left(\repX^{\alpha}_{s_1,r_1}\tensor\repX^{\beta}_{1,r_2+n}\right)\tensor\modNbar^-_{p-s_2,n}(1),\\
&\modN^{\alpha}_{s_1,r_1}(n_1)\tensor\modN^{\beta}_{s_2,r_2}(n_2) = \left(\repX^{\alpha}_{1,n_1}\tensor\repX^{\beta}_{1,n_2}\right)\tensor\left(\modN^+_{s_1,r_1+n_1-1}(1)\tensor\modN^+_{s_2,r_2+n_2-1}(1)\right),\\
&\modN^{\alpha}_{s_1,r_1}(n_1)\tensor\modNbar^{\beta}_{s_2,r_2}(n_2) = \left(\repX^{\alpha}_{1,n_1}\tensor\repX^{\beta}_{1,n_2}\right)\tensor\left(\modN^+_{s_1,r_1+n_1-1}(1)\tensor\modNbar^+_{s_2,r_2+n_2-1}(1)\right),\\
&\modN^{\alpha}_{s_1,r_1}(n_1)\tensor\modW^{\beta}_{s_2,r_2}(n_2) = \left(\repX^{\alpha}_{1,n_1}\tensor\repX^{\beta}_{1,r_2+n_2}\right)\tensor\left(\modN^+_{s_1,r_1+n_1-1}(1)\tensor\modN^-_{p-s_2,n_2}(1)\right),\\
&\modN^{\alpha}_{s_1,r_1}(n_1)\tensor\modM^{\beta}_{s_2,r_2}(n_2) = \left(\repX^{\alpha}_{1,n_1}\tensor\repX^{\beta}_{1,r_2+n_2}\right)\tensor\left(\modN^+_{s_1,r_1+n_1-1}(1)\tensor\modNbar^-_{p-s_2,n_2}(1)\right),\\
&\modNbar^{\alpha}_{s_1,r_1}(n_1)\tensor\modNbar^{\beta}_{s_2,r_2}(n_2) = \left(\repX^{\alpha}_{1,n_1}\tensor\repX^{\beta}_{1,n_2}\right)\tensor\left(\modNbar^+_{s_1,r_1+n_1-1}(1)\tensor\modNbar^+_{s_2,r_2+n_2-1}(1)\right),\\
&\modNbar^{\alpha}_{s_1,r_1}(n_1)\tensor\modW^{\beta}_{s_2,r_2}(n_2) = \left(\repX^{\alpha}_{1,n_1}\tensor\repX^{\beta}_{1,r_2+n_2}\right)\tensor\left(\modNbar^+_{s_1,r_1+n_1-1}(1)\tensor\modN^-_{p-s_2,n_2}(1)\right),\\
&\modNbar^{\alpha}_{s_1,r_1}(n_1)\tensor\modM^{\beta}_{s_2,r_2}(n_2) = \left(\repX^{\alpha}_{1,n_1}\tensor\repX^{\beta}_{1,r_2+n_2}\right)\tensor\left(\modNbar^+_{s_1,r_1+n_1-1}(1)\tensor\modNbar^-_{p-s_2,n_2}(1)\right),\\
&\modW^{\alpha}_{s_1,r_1}(n_1)\tensor\modW^{\beta}_{s_2,r_2}(n_2) = \left(\repX^{\alpha}_{1,r_1+n_1}\tensor\repX^{\beta}_{1,r_2+n_2}\right)\tensor\left(\modN^-_{p-s_1,n_1}(1)\tensor\modN^-_{p-s_2,n_2}(1)\right),\\
&\modW^{\alpha}_{s_1,r_1}(n_1)\tensor\modM^{\beta}_{s_2,r_2}(n_2) = \left(\repX^{\alpha}_{1,r_1+n_1}\tensor\repX^{\beta}_{1,r_2+n_2}\right)\tensor\left(\modN^-_{p-s_1,n_1}(1)\tensor\modNbar^-_{p-s_2,n_2}(1)\right),\\
&\modM^{\alpha}_{s_1,r_1}(n_1)\tensor\modM^{\beta}_{s_2,r_2}(n_2) = \left(\repX^{\alpha}_{1,r_1+n_1}\tensor\repX^{\beta}_{1,r_2+n_2}\right)\tensor\left(\modNbar^-_{p-s_1,n_1}(1)\tensor\modNbar^-_{p-s_2,n_2}(1)\right).
\end{align*}
Explicit decompositions of these tensor products easily follow by
applying~\eqref{irred-fusion} and Thm.~\bref{tens-prod-decomp-1},
and Thm.~\bref{tens-prod-decomp-2}. 

\medskip

\section{The $A_N$ quivers}\label{app:quivers}
We here recall basic notions about quivers~\cite{[CB],[ARS]}.

\subsection{Quivers and their representations}\label{sec:quivers}
A quiver is an oriented graph, that is, a quadruple $(I, A, s, t)$,
consisting of a finite set $I$ of vertices, a finite set $A$ of
oriented edges (arrows), and two maps $s$ and $t$ from $A$ to $I$.  An
oriented edge $a \in A$ starts at the vertex $s(a)$ and terminates at
$t(a)$.

A \textit{representation} of a quiver $Q$ (over $\oC$) is a collection
of finite-dimensional vector spaces $V_i$ over $\oC$, one for each
vertex~$i \in I$ of~$Q$, and $\oC$-linear maps $f_{ij}: V_i \to V_j$,
one for each oriented edge $\stackrel{i}{\bullet}\xrightarrow{a}
\stackrel{j}{\bullet}$.  \ The \textit{dimension} of a representation
$\rho$ of $Q$ is an element of $\oZ[I]$ given by the dimensions of
$V_i$, $i\in I$: $\dim(\rho)= \sum_{\substack{i\in
    I}}\dim_{\oC}(V_i)i$.

A \textit{morphism} from a representation $\rho$ of a quiver $Q$ to
another representation $\rho'$ of $Q$ is an $I$-graded $\oC$-linear
map $\phi=\bigoplus_{i\in I}\phi_i:\bigoplus_{i\in
  I}V_{i}\to\bigoplus_{i\in I}V_{i}'$ satisfying $f_{ij}'\phi_i=\phi_j
f_{ij}$ for each oriented edge $\stackrel{i}{\bullet}\xrightarrow{a}
\stackrel{j}{\bullet}$.  This gives the category of representations of
the quiver $Q$, to be denoted by $\Rep(Q)$ in what follows.

If a quiver $Q$ has no oriented cycles, isomorphism classes of simple
objects in $\Rep(Q)$ are in a one-to-one correspondence with vertices
of $Q$.  The simple object corresponding to a vertex $i\in I$ is given
by the vector spaces
\begin{align*}
  V_j&=
 \begin{cases}
   \oC,\quad j\!=\!i,\\
   0,\quad \text{otherwise}
 \end{cases}\\
 \intertext{and $\oC$-linear maps}
 f_{ij}&=0,\quad \text{for all}\; i,j\in I.
\end{align*}
A quiver is said to be of \textit{finite} type if the underlying
nonoriented graph is a Dynkin graph of finite type. 
A quiver is said
to be of \textit{simply laced} type if it does not have a pair of
vertices connected by more than one arrow.

\subsection{The category $\Rep(\qK)$ of representations of the
  $A_N$ quiver}\label{subsec:rep-AN} 
The $A_N$ quiver $\qK_N$ is simply-laced and has $N$ vertices connected by $N-1$ single edges,
 that is,
 $\qK_N =
 (\{1,2,\dots,N\},\linebreak[0]\{g_i\},s,t)$, where $s(g_i)=i+1$ and $t(g_i)=i$ for odd $i$, and  $s(g_i)=i$ and $t(g_i)=i+1$ for even $i$:
 \begin{equation*}
 \qK_N:\qquad
   \xymatrix{
     {\stackrel{1}{\bullet}}
     &{\stackrel{2}{\bullet}}\ar@<+0.8ex>[l]_{g_1}\ar@<-0.8ex>[r]^{g_2} 
     &{\stackrel{3}{\bullet}}
     &{\stackrel{4}{\bullet}}\ar@<+0.8ex>[l]_{g_3}\ar@<-0.8ex>[r]^{g_4} 
     &\dots
     &{\stackrel{N}{\bullet}}\ar@<+0.8ex>[l]_{g_{N-1}} 
   }
 \end{equation*}
 where we also assume that $N$ is even, for simplicity.
 A representation $\rho$ of $\qK_N$ is a collection
 $((V_1,V_2,\dots,V_N),(f_{2,1},f_{2,3},f_{4,3},f_{4,5},\dots,f_{N,N-1}))$ consisting of $N$
 vector
 spaces $V_j$ and $N-1$ linear maps $f_{j,j\pm1}\in
 \mathrm{Hom}_{\oC}(V_j,V_{j\pm1})$, where $2\leq j\leq N$ is even.  The dimension of~$\rho$ is given by
 $\dim(\rho)=(\dim_{\oC}(V_1), \dim_{\oC}(V_2), \dots, \dim_{\oC}(V_N))$.  Simple objects in
 the category $\Rep(\qK_N)$ are given by the $N$ representations
 $\rho_r=((\delta_{ir}\oC),(0,\dots,0))$.
  We now
recall the classification of indecomposable representations of the
$A_N$ quiver $\qK_N$, summarized in~\bref{prop:repr-AN} below.

There is a correspondence between indecomposable representations of a
quiver and the set $\Delta_{+}$ of positive roots of the Lie algebra
corresponding to the Dynkin graph associated with the quiver.  This
correspondence is one-to-one for a quiver of simply laced finite type
\cite{[Gab],[BGP]}.  Namely, a representation $\rho$ of a quiver $Q$
is indecomposable if and only if $\dim(\rho)\in\Delta_{+}$ and,
conversely, for every $\alpha\in\Delta_{+}$, there is, up to an
isomorphism, a unique indecomposable representation $\rho$ of the
quiver $Q$ such that $\dim(\rho)=\alpha$.

The nonoriented graph associated with the quiver $\qK_N$ is
the finite Dynkin graph $A_N$.  It is well known that
$\alpha\in\Delta_+$ is a positive root of $A_N$ if
$\alpha=(\alpha_i)$ with $\alpha_i\in\{0,1\}$ and at least one $\alpha_i=1$ for $1\leq i\leq N$.
 In particular, $\alpha_r=(\delta_{ir})$, for $1\leq r\leq N$, are the simple roots.
 The simple roots $\alpha_r$ correspond to the respective simple objects
$\rho_r$ in the
category $\Rep(\qK_N)$.  The other positive roots $\alpha$ are in a one-to-one
correspondence with indecomposable representations of dimension $\alpha$:
$\rho(\alpha)=((\alpha_i\oC),(f_{2,1},f_{2,3},f_{4,3}\dots))$ with maps $f_{i,i\pm1}$ defined in an obvious way. We summarize these results in the
following well-known proposition (see, e.g., \cite{[FMV]}).

\begin{prop}\label{prop:repr-AN}\mbox{}
  \begin{enumerate}
    
  \item If $\alpha\notin\Delta_{+}$, then the set of indecomposable
    representations of $\qK_N$ with the dimension $\alpha$ is
    empty.\pagebreak[3]
    
  \item If $\alpha\!\in\!\Delta_{+}$, then an
    indecomposable representation of $\qK_N$ with the dimension
    $\alpha$ is either the representation $\rho(\alpha)=((\alpha_i\oC),(f_{2,1},f_{2,3},f_{4,3}\dots))$ with maps $f_{i,i\pm1}$ defined in an obvious way, where $\alpha=(\alpha_i)$.
 \end{enumerate}
\end{prop}

\end{document}